\documentclass{IEEEtran}

\usepackage{amsfonts}
\usepackage{mathrsfs}
\usepackage{amssymb,amsmath}
\usepackage[noblocks]{authblk}
\usepackage[compress]{cite}
\usepackage{bm}
\usepackage{algorithm}
\usepackage{algorithmic}
\usepackage{amsmath,amssymb,epsfig,graphics,subfigure}
\usepackage{theorem}
\usepackage{array,color}
\usepackage[compress]{cite}
\usepackage{algorithm}
\usepackage{algorithmic}
\usepackage[OT1]{fontenc}

\allowdisplaybreaks

\newtheorem{lemma}{Lemma}
\newtheorem{proof}{Proof}

\theoremheaderfont{\normalfont\bfseries}

\begin{document}

\title{New Viewpoint and Algorithms for Water-Filling Solutions in Wireless Communications}

\author{Chengwen Xing, Yindi Jing, Shuai Wang, Shaodan Ma, and H. Vincent Poor, \textsl{Fellow}, \textsl{IEEE} \thanks{C. Xing and S. Wang are with School of Information and Electronics, Beijing Institute of Technology, Beijing 100081,
China (e-mail: chengwenxing@ieee.org and swang@bit.edu.cn). }  \thanks{Y. Jing is with the Department of Electrical and Computer Engineering,
University of Alberta, Edmonton, AB T6G 1H7, Canada (e-mail: yindi@ualberta.ca). }
\thanks{S. Ma is with the State Key Laboratory of Internet of Things for Smart City and the Department of Electrical and Computer Engineering, University of Macau, Macao SAR, China(e-mail: shaodanma@um.edu.mo).}
\thanks{H. V. Poor is with the Department of Electrical Engineering, Princeton University, Princeton, NJ 08544, USA. (e-mail:
poor@princeton.edu).}
} \maketitle

\begin{abstract}
Water-filling solutions play an important role in  the designs for wireless communications, e.g., transmit covariance matrix design. A traditional physical understanding is to use the analogy of pouring water over a pool with fluctuating bottom. Numerous variants of water-filling solutions have been discovered during the evolution of wireless networks. To obtain the solution  values, iterative computations are required, even for simple cases with compact mathematical formulations. Thus, algorithm design is a key issue for the practical use of water-filling solutions, which however has been given marginal attention in the literature. Many existing algorithms are designed on a case-by-case basis for the variations of water-filling solutions  and/or with complex logics. In this paper, a new viewpoint for water-filling solutions is proposed to understand the problem \textit{dynamically} by considering changes in the increasing rates on different subchannels. This fresh viewpoint provides useful mechanism and fundamental information in finding the optimization solution values. Based on the new understanding, a novel and comprehensive method for practical water-filling algorithm design is proposed, which can be used for systems with various performance metrics and power constraints,  even for systems with imperfect channel state information (CSI). 

\end{abstract}


\section{Introduction}



Water-filling solutions play a central role in the optimization of communication systems. They are undoubtedly among the most fundamental and important results in wireless communication designs, signal processing designs, and network optimizations including transceiver optimization, training optimization, resource allocation, and so on, e.g.,  \cite{Cover2006,Boyd04,Goldsmith2005,Goldsmith1996,JYang1994,Waterfilling2004,Jafar2005,Teletar1995,Sampth01,Sampth03,Tse2005,Palomar_Waterfilling,Dai2014,Ozel2014,Palomar03}. Loosely speaking, optimal resource allocations for multi-dimensional communication systems usually lead to water-filling solutions. Over the past decade, wireless systems have evolved dramatically  and exhibited a great variety of configurations with many different performance requirements and physical constraints, e.g., \cite{Lifeng_Lai}.
This diversity results in a rich body of variants of water-filling solutions \cite{Khakurel2014,Jafar2005,Sampth03,Waterfilling2004,HMoon2011,Gao200902,PeterHe2013,PeterHe2017,
PeterHeTWC2018,PeterHeTSG2018,PeterHe2015,Palomar_Waterfilling,Sampth01,Dai2014,Lifeng_Lai,ClusterWF,Ling2012,Popescu2007,Hoang2008,Palomar03,Ozel2014}, from single water-level ones to multiple water-level ones \cite{Palomar_Waterfilling}, from solutions for perfect channel state information (CSI) to robust ones such as cluster water-filling \cite{ClusterWF}, and from constant water-level ones to cave-filling ones \cite{FFGaoTSP}.

In many conventional works, the first step in obtaining a water-filling solution for an optimization problem is to find the Karush-Kuhn-Tucker (KKT) conditions and manipulate them into a recognizable format which is usually referred to as the water-filling solution. KKT conditions are necessary conditions for the optimization, and are also sufficient if the problem is convex \cite{Boyd04}. While KKT conditions determine the optimal solutions, their initial formats are implicit and do not provide information in how to achieve the optimal solution values. Thus, sophisticated mathematical manipulations are needed to transform them into a water-filling structure. As communication systems and optimization problems get more complicated, the corresponding KKT conditions also become more  complicated, both in mathematical complexity and in the number of equations. Manipulating the KKT conditions into a recognizable format may become very difficult. First, the large number of complicated KKT conditions hinder efficient manipulations and clear understanding of their physical meaning in terms of  water-filling structure. Moreover, the derived water-filling solutions may not have compact and systematic format to allow the development of water-filling algorithms in an effective and unified manner.


Furthermore, the optimization design is not complete with the derived water-filling solutions as the solutions contain unknown parameters such as water levels. In other words, the solutions are still in  \textit{implicit form}. Thus, an important second step in obtaining the water-filling solution of an optimization problem is to find a practical algorithm. This step has not been given sufficient attention and in some cases has been ignored. Generally speaking, water-filling solutions consist of two major components, i.e., water level and water bottom, and a traditional imagery of puring is to pour water over a pool with different bottoms \cite{Cover2006}. Based on this analogy, several practical water-filling algorithms have been proposed \cite{PeterHe2013,Antonio2014,PeterHe2014,PeterHeTWC2018,
PeterHeTSG2018,PeterHeTVT2018,PeterHe2015,Bjornson2010,Dai2012,
Katoh1979,PeterHe2017,FFGaoTSP}. They generally differ from each other in many perspectives, e.g., optimization objectives and constraints, design logics, algorithm structures, computational complexity and so on. Usually water-filling algorithms are designed on a case-by-case basis. Interestingly, an efficient water-filling design framework is proposed in \cite{PeterHe2013} based on a geometric understanding of the water-filling operation. It generally applies to the throughput maximization under various constraints.


In this paper, we provide a new viewpoint on water-filling solutions. It has three major advantages: 1) it helps the understanding of water-filling results; 2) it avoids tedious and challenging manipulations of KKT conditions; and 3) it leads to efficient algorithms to find the solution values. Based on this new understanding, a unified water-filling algorithm design framework is proposed from an algebraic viewpoint instead of the geometric viewpoint in \cite{PeterHe2013}. The unified framework is general and applicable for various complicated communication optimization problems including but not limited to throughput maximization and mean-squared-error (MSE) minimization under general power constraints. The main contributions are summarized as follows.


\begin{itemize}
\item We provide a novel understanding from a dynamic perspective for optimization problems with water-filling solutions. In contrast with the traditional approach, this viewpoint can avoid tedious manipulations of KKT conditions in deriving water-filling solutions and greatly simplify water-filling algorithm design.

\item A standard and plausible notation used in water-filling solutions is the ``+'' operation where $x^+\triangleq\max\{x,0\}$. Its widely acknowledged physical meaning is that the resource (e.g., power) allocated to a subchannel must be nonnegative. However, direct implementation of the ``+'' operation may result in inconvenience and lead to problematic solutions for some optimization problems. In our work, index based operations are introduced in the algorithm designs to avoid the ``+'' operation and simplify the algorithm design.

\item  In addition to efficiency, the proposed method and the resulting algorithms are highly intuitive and understandable, and are also attractive from the implementation perspective. It is also compatible for extensions to complicated systems by using  simple cases as building blocks.

\item With the proposed method, we investigate a class of communication optimizations  with general convex objective functions under box constraints, where the allocated resource of each subchannel is bounded from both ends.
Corresponding algorithms for the optimal solution values are proposed.  Moreover, the algorithms can be extended  to serve even more general problems and have a wide range of applications.

\item  Robust optimizations for wireless systems with CSI uncertainties are also studied.
Algorithms for finding the solutions are proposed for robust weighted-MSE minimization, robust capacity maximization, robust worst-MSE minimization, and robust minimum capacity maximization for multiple-input multiple-output (MIMO) orthogonal frequency-division multiplexing (OFDM)\ systems, the last two of which were largely open.
\end{itemize}


\section{A New Viewpoint of Water-filling Solutions}
\label{sec_new_understanding}

We consider a convex optimization problem of the following form:
\begin{align}
\label{Opt_1}
\text{P1}: \ \max_{p_1,\cdots,p_K} \ \ &{\sum}_{k=1}^K f_k(p_k), \nonumber \\
\ {\rm{s.t.}} \ \ &{\sum}_{k=1}^K p_k \le P, \ \ p_k\ge0,
\end{align}where $P>0$ and the functions $f_k(\cdot)$ are real-valued, increasing, and strictly concave. Further assume that $f'_k(\cdot)$'s are continuous, where $f'_k(\cdot)$ denotes the first order derivative of $f_k(\cdot)$.
Many optimization problems in wireless communications have this format or contain this problem as an essential part, for example, the power allocation problem in MIMO capacity maximization. It is known that the optimal solution of (\ref{Opt_1}) has a water-filling structure. In what follows, we first explain the traditional treatment of this problem, then our new viewpoint and algorithm are elaborated along with the comparison of the two algorithms.\vspace{-3mm}

\subsection{Existing Treatment for Water-Filling Solutions}
Traditionally, Lagrange multiplier method has been used for (\ref{Opt_1}). The first step is to find the KKT conditions and from them to derive the water-filling solution of the problem in a compact format. As the objective function is a sum of decomposed concave functions and the constraints are linear, the problem is a convex one. Thus the KKT conditions are both necessary and sufficient. With straightforward calculations, the KKT conditions of (\ref{Opt_1}) are
\begin{align}
\label{eq2}
&f_k'(p_k)=\mu-\lambda_k,\
\mu\left({\sum}_{k=1}^K p_k - P\right)=0, \ \lambda_kp_k=0,
\end{align}where $\mu$ and $\lambda_k$ are the Lagrange multipliers corresponding to the two constraint sets. By rewriting the KKT conditions, the solution has the following water-filling structure:
\begin{eqnarray}
p_k=\left[g_k(\mu)\right]^+,\quad \text{and} \quad{\sum}_{k=1}^K p_k = P,
\label{water_c_3}
\end{eqnarray}
where
\begin{eqnarray}
g_k(\cdot)\triangleq {\rm Inv}[f_k'](\cdot),
\label{gk}
\end{eqnarray}
i.e., $g_k(\cdot)$ is the inverse function of $f_k'(\cdot)$. As $f_k(\cdot)$ is increasing and concave, the derivative function $f_k'(\cdot)$ is positive and monotonically decreasing. Based on the definition $f_k'(p_k)=\mu_k$, the inverse function $g_k(\cdot)$ denotes the mapping relationship from $\mu_k$ to $p_k$, i.e., $p_k=g_k(\mu_k)$. 
The Lagrange multiplier $\mu$ has the physical meaning of the water level. On the other hand, the function of $\lambda_k$'s is implicit in this water-filling solution as they affect the solution through the ``+'' operation. We would like to highlight that the ``+'' operation results from rigorous mathematical derivations. While it can be explained intuitively by ``the power level must be nonnegative,'' the ``+'' operation should not be added recklessly during the derivations merely due to this physical meaning. For more general problems, such practice can lead to sub-optimality in the solution.

Another important step in using water-filling solutions in communications systems is to obtain the solution values, i.e., the values of $p_k$'s of the solution in (\ref{water_c_3}). It is a non-trivial step. Existing algorithms are usually for specific applications and a unified framework is missing.

To obtain the values of the $p_k$'s from (\ref{water_c_3}), a practical water-filling algorithm is needed.  The major challenge is to find the index set of active subchannels with non-zero powers, i.e.,
\begin{equation}
\mathcal{S}_{active}=\{k|p_k>0\}.
\end{equation}In general, all possible subchannel combinations need to be studied. For each of the $2^K-1$ possibilities, the corresponding $p_k$-values can be found, and then the one with the highest objective function value among the $2^K-1$ possibilities is the optimal solution.
But this method is obviously inefficient as the complexity is exponential in $K$. For settings with simple $f_k$-functions and fortunate parameter values, a natural ordering of the subchannels exists, and the algorithm can be designed to have a lower complexity, where the number of possible active sets to be explored has the order of $\mathcal{O}(K\log K)$.

In \cite{Katoh1979}, the class of optimization problems P1 was studied for the case that the optimization variables takes only non-negative integer values. By discretizing the continuous range of the variables in P1 into a fine grid, the algorithm proposed in [32] may be used to find an approximate solution to P1, which yields to a scheme similar to gradient-based search. But in transforming the continuous problem into a discrete one, the precision is always a concern. The discretization error depends on both the step size and $f_k$-functions.
\vspace{-2mm}


\subsection{New Viewpoint and Algorithm}
\label{sec-2-2}

The traditional method for P1 as explained in the previous subsection has two major disadvantages. The first is the need of the transformation from KKT conditions to water-filling solutions. As the problem gets more general for more involved wireless systems and models, the transformation can become intractable. The second is the lack of general and effective algorithms in finding the values of the solution. In the following, from the perspective of a dynamic procedure, we give a new viewpoint on the solution of the optimization problem, which helps address both challenges. Since $f_k(\cdot)$ is concave, $f_k'(\cdot)$ is a decreasing function meaning that the increasing rate of $f_k(p_k)$ decreases as $p_k$ increases. The optimization problem P1 aims at allocating the total power $P$ over a series of functions, i.e., $f_k(p_k)$'s. We can  see this problem as dividing the available power $P$ into a large number of small portions and the power is allocated portion by portion. For each portion, we should choose the subchannel whose $f_k(\cdot)$ has the maximum increasing rate to maximize the total of $f_k(\cdot)$'s. As the increasing rate of this $f_k$-function decreases when a resource portion is added to it, after getting a certain amount of power portions, its increasing rate may become smaller than another subchannel. In this case, a new subchannel will have the fastest increasing rate and the next power portion should be added to this new subchannel. This procedure repeats until all resource portions have been allocated.
When the resource allocation stops, the functions that are allocated with nonzero powers will have the same increasing rate. Some subchannels may never get any power portion if their increasing rates are never the highest.

The result discussed above is presented in the following claim with rigorous proof.
\begin{lemma}
\label{water_1}
The following conditions are both necessary and sufficient for the optimal solution of P1:
\begin{align}
\left\{\begin{array}{ll} f_k'(p_k)=f_j'(p_j) & \text{for } k,j\in\mathcal{S}_{active}; \\
f_j'(p_j=0)\le f_k'(p_k) & \text{for } k\in\mathcal{S}_{active} \text{ and } j\notin\mathcal{S}_{active}; \\
\sum_{k=1}^K p_k=P. &
\end{array} \right.
\label{eqn-6}
\end{align}
\label{lemma-1}
\end{lemma}
\vspace{-8mm}

\begin{proof} We first prove the necessity part by contradiction. The necessity of the last line of (\ref{eqn-6}) is obvious and has been proved in many existing work. Thus the proof is omitted here. Denote the optimal solution of P1 as $p_1^*,\cdots,p_K^*$.  Assume without loss of generality that $p_1^*,p_2^*>0$ (i.e., $1,2\in \mathcal{S}_{active}$) but $f_1'(p^*_1)>f_2'(p^*_2)$. Since $f_1'$ and $f_2'$ are continuous, there exists an $\delta$ with $0<\delta<p_2^*$ such that $f_1'(p^*_1+x)>f_2'(p^*_2-x)$ for $0< x\le \delta$. Thus
\begin{align}
&f_1(p^*_1+\delta)+f_2(p^*_2-\delta) \nonumber  \\
&=f_1(p^*_1)+f_2(p^*_2)+\int_{0}^{\delta}
\left[f_1'(p^*_1+x) -f_2'(p^*_2-x)\right]dx \nonumber \\
&>f_1(p^*_1)+f_2(p^*_2).
\end{align}This shows that the new solution $\{p_1^*+\delta, p_2^*-\delta,p_3^*,\cdots,p_K^*\}$ (which satisfies all constraints by construction) is better than $p_1^*, p_2^*,p_3^*,\cdots,p_K^*$, which contradicts the assumption. This proves that the first line of (\ref{eqn-6}) is necessary.

Similarly, to prove that the second line of (\ref{eqn-6}) is necessary, assume without loss of generality that $p_1^*>0,p_2^*=0$ (i.e., $1\in \mathcal{S}_{active}$ and $2\notin \mathcal{S}_{active}$) but $f_2'(0)>f_1'(p^*_1) $. Since $f_2$ is strictly concave and $f_2'$ is continuous, there exists an $\delta$ with $0<\delta<p_1^*$ such that $f_1'(p^*_1-x)<f_2'(x)$ for $0< x\le \delta$. Thus
\begin{align}
&f_1(p^*_1-\delta)+f_2(\delta)\nonumber \\
&=f_1(p^*_1)+f_2(0)+\int_{0}^{\delta}
\left[f_2'(x) -f_1'(p^*_1-x)\right]{\rm{d}}x\nonumber \\
&>f_1(p^*_1)+f_2(0).
\end{align}This says that the solution $\{p_1^*-\delta, \delta,p_3^*,\cdots,p_K^*\}$ is better and thus leads to a contradiction.

For the sufficiency, it is enough to show that a solution satisfying (\ref{eqn-6}) is a local maximum. Since P1 is a convex optimization, its local maximum is unique and is the global maximum. Let $\{p_1^*,\cdots,p_K^*\}$ be the solution satisfying $(\ref{eqn-6})$ and consider a solution $\{p_1,\cdots,p_K\}$ in the vicinity of it. Define $\mathcal{S}_{+}\triangleq\{k|p_k>p_k^*\}$ and $\mathcal{S}_{-}\triangleq\{k|p_k<p_k^*\}$.  Notice that $\mathcal{S}_{-}\cap \mathcal{S}_{inactive}^*=\emptyset$, where $\mathcal{S}_{inactive}^*=\{k|p_k^*=0\}$. Thus
\begin{align}
&{\sum}_{k=1}^K f_k(p_k)={\sum}_{k=1}^K f_k(p_k^*) +\sum_{k\in \mathcal{S}_{+}}\!\!
\int_{0}^{p_k-p_k^*}\!\!\! f_k'(p_k^*+x_k)dx_k \nonumber \\
&\hspace{2cm}- \sum_{\hat{k}\in \mathcal{S}_{-}}
\!\!\!{\int}_{0}^{p_{\hat{k}}^*-p_{\hat{k}}}\!\!\! f_{\hat{k}}'(p_{\hat{k}}^*-\hat{x}_{\hat{k}})
d\hat{x}_{\hat{k}}.
\end{align}From the conditions on $f_k$'s and the assumption that $\{p_1^*,\cdots,p_K^*\}$ satisfies (\ref{eqn-6}), we have
\begin{equation}
f_k'(p_k^*+x)<f_k'(p_k^*)\le f_{\hat{k}}'(p_{\hat{k}}^*)<f_{\hat{k}}'(p_{\hat{k}}^*-\hat{x}),
\label{proof-cond-1}
\end{equation}for all $k\in \mathcal{S}_{+}$, $\hat{k}\in \mathcal{S}_{-}$, $x_k\in(0,p_k-p_k^*)$, $\hat{x}_{\hat{k}}\in (0, p_{\hat{k}}^*-p_{\hat{k}})$. Also, since $\sum_{k=1}^K p_k^*=P\ge \sum_{k=1}^K p_k$, we have
\begin{equation}
{\sum}_{k\in \mathcal{S}_{+}} (p_k-p_k^*)\le {\sum}_{\hat{k}\in \mathcal{S}_{-}}
(p_{\hat{k}}^*-p_{\hat{k}}).
\label{proof-cond-2}
\end{equation}By combining (\ref{proof-cond-1}) and  (\ref{proof-cond-2}), it can be concluded that $\sum_{k=1}^K f_k(p_k)<\sum_{k=1}^K f_k(p_k^*)$, and thus $\{p_1^*,\cdots,p_K^*\}$ is a local maximum.\footnote{The lemma can also be proved by showing that (6) is equivalent to the KKT conditions, which are necessary and sufficient for P1. But here we use a direct proof to help illustrate the proposed new viewpoint and avoid unnecessary dependence on existing water-filling results.}
\end{proof}

From (\ref{eqn-6}), we see that the value of $f_k'(p_k)$ for $k\in \mathcal{S}_{active}$, denoted as $\mu$, is the increasing rate for the optimal power allocation result. The allocated power on the subchannels can also be represented as functions of $\mu$:
\begin{align}
\label{water_3}
\left\{\begin{array}{ll}p_k=g_k(\mu) & \text{ for $k\in \mathcal{S}_{active}$}\\ p_k=0 & \text{ for $k\notin \mathcal{S}_{active}$} \end{array}\right.,
\end{align}where $g_k(\cdot)$ is defined in (\ref{gk}). From the total power constraint,
\begin{align}
\label{water_4}
P={\sum}_{k \in \mathcal{S}_{active}}g_k(\mu),
\end{align}based on which $\mu$ can be solved when the set of active subchannels $\mathcal{S}_{active}$ is known.

As explained in the previous subsection. The main difficulty of finding the solution values is to find $\mathcal{S}_{active}$. We propose the use of index operations ${\mathcal{I}}_k$'s to conquer this difficulty.
When subchannel $k$ is allocated nonzero power, ${\mathcal{I}}_k=1$, otherwise ${\mathcal{I}}_k=0$. With these indices, (\ref{water_3}) and (\ref{water_4}) are rewritten as
\begin{align}
\label{water_c_1}
p_k=g_k(\mu){\mathcal{I}}_k, \ k=1,\cdots,K,  \text{ and } P={\sum}_{k=1}^Kg_k(\mu){\mathcal{I}}_k.
\end{align} Clearly, for the subchannels in the inactive set, their corresponding indices and powers are set as zero, i.e., $\mathcal{I}_k=0$ and $p_k=0$. They are not necessary to be involved in the calculation in (14). It is noteworthy that $g_k(\cdot)$ is used to denote the inverse mapping of $f_k'(\cdot)$. An explicit closed-form or analytical expression for $g_k(\cdot)$ is not necessary. For cases that $g_k(\mu)$ cannot be written in an explicit formula, (\ref{water_c_1}) can be understood by the following alternative form
\begin{align}
\label{water_c_alter}
f_k'(p_k)=\mu \ \text{for} \ {\mathcal{I}}_k=1 \ \text{and} \  {\sum}_{k=1}^Kp_k=P.
\end{align}
With this result, we present a water-filling algorithm for P1 in \textbf{Algorithm~\ref{algorithm_1}}.
\begin{algorithm}[t]
\caption{Proposed water-filling algorithm for P1.}
\label{algorithm_1}
\begin{algorithmic}[1]
\STATE $\mathcal{I}_k=1$  $\text{for}$ $k=1,\cdots,K;$
\STATE Calculate $\mu$ and $p_k$'s $ \text{using} $  Eqn.~(\ref{water_c_1});
\WHILE{$\text{length}( \text{find}(\{p_k\} < 0 )) > 0$}
\STATE Find $\mathcal{S}_{inactive}\hspace{-1mm}=\{k|p_k\le0\}$;
\STATE Set $\mathcal{I}_k=0$ for $k\in \mathcal{S}_{inactive}$;
\STATE  $\text{Calculate}$ $\mu$ and ${p_k}$'s using Eqn.~(\ref{water_c_1});
\ENDWHILE
\RETURN $p_k$'s
\end{algorithmic}
\end{algorithm}In the first step of \textbf{Algorithm~\ref{algorithm_1}}, all subchannels are initialized as active and in the second step, the corresponding increasing rate and subchannel powers are calculated. As the computations of $p_k$'s do not consider the constraints that $p_k\ge 0$, it may appear that $p_k<0$ for some $k$. In this case, the corresponding index ${\mathcal{I}}_k$ will be set to zero and this subchannel will be allocated zero-power in the next round. In other words, this subchannel is inactive and there is no need to compute the corresponding derivative and inverse functions, i.e., $f_k'(p_k=0)$.  The procedure continues until all active subchannels are allocated nonnegative powers.


\begin{lemma}
Algorithm~\ref{algorithm_1} converges and achieves the optimal solution of P1.
\label{lemma-2}
\end{lemma}

\begin{proof} Since for each iteration in Algorithm~\ref{algorithm_1}, the new set for $\mathcal{S}_{active}=\{k|p_k>0\}$ is either the same as the previous $\mathcal{S}_{active}$ (thus the algorithm terminates) or shrinks to a subset of the previous $\mathcal{S}_{active}$. As the size of the initial set is $K$, it is obvious that the algorithm converges within $K$ iterations.

Now we prove that Algorithm~\ref{algorithm_1} converges to the optimal solution of P1. First, since $P>0$, at any iteration, it is impossible to have $p_k\le0$ for all $k$. In other words, there exists a $k$ such that $p_k>0$. Let $\{p_1,\cdots,p_K\}$ be the solution found by Algorithm~\ref{algorithm_1} at the $m$th iteration. From  Step 2 and Step 6, it is obvious that the solution satisfies the first and last conditions of (\ref{eqn-6}). For any $j \notin \mathcal{S}_{active}=\{k|p_k>0\}$, we have $p_j<0$ in  one of the previous iterations. Denote the iteration round for $p_j < 0$ as $m'$. Thus from (\ref{water_c_1}), $p_j=g_j(\mu^{(m')})<0$, from which $\mu^{(m')}>f_j'(0)$, where $\mu^{(m')}$ is the achieved increasing rate at the $m'$th iteration. Notice that $\mu^{(m')}=f_k'(p_k^{(m')})$ for subchannel $k$ in the active set of the $m'$th iteration.  With the proposed algorithm, subchannel $j$ is removed by setting $p_j=0$, and in the next iteration, the sum power available for the remaining active subchannels decreases. The achieved increasing rate for this new iteration is higher, i.e., $\mu^{(m')}<\mu^{(m'+1)}$. Denote the overall iteration number for the algorithm as $m$. Since $m'\le m$, we have $f_j'(0)<\mu^{(m')}\le \mu^{(m)}=f_k'(p_k)$ for $k\in\mathcal{S}_{active}$. This proves that the solution found by the algorithm also satisfies the second condition of (\ref{eqn-6}). As (\ref{eqn-6}) is proved to be sufficient for the optimal solution in Lemma \ref{lemma-1}, the solution found by Algorithm~\ref{algorithm_1} is thus the optimal one.\end{proof}

\noindent \textbf{Remark:} When the inverse functions in (\ref{gk}) and $\mu$ can be derived in closed-forms, the water-filling solution and Algorithm 1 can be implemented efficiently. For each iteration of Algorithm 1, the complexity of the calculations of $\mu$ and $p_k$ is $\mathcal{O}(K)$. Since there are at most $K-1$ iterations, the worst-case complexity of Algorithm 1 is ${\mathcal{O}}(K^2)$.

Otherwise when the inverse function or $\mu$ does not have a closed-form, a numerical method such as bisection search is needed for an approximate solution. The complexity of Algorithm 1 depends on the numerical algorithm and precision.
As $f_k'(p_k)$ is monotonic in $p_k$ and together with the facts that $f_k'(p_k)=\mu$ for ${\mathcal{I}}_k=1$ and $\sum_{k=1}^Kp_k=P$, the variables $p_k$'s can always be computed by using two dimensional bisection search. More specifically, the inner round bisection search is performed over $p_k$'s regarding to the equalities $f_k'(p_k)=\mu$ for $k=1,\cdots,K$ and the outer round bisection search is performed over $\mu$ considering the constraint $\sum_{k=1}^Kp_k=P$.
The inner bisection algorithm is given by Algorithm~\ref{alg-bisection-inner} and the two-dimensional bisection search algorithm is given by Algorithm~\ref{alg-bisection-two}.
It is worth highlighting that if we have $p_k<0$ by solving $f_k'(p_k)=\mu$, it means that for $p_k\ge0$, $f_k'(p_k)=\mu$ cannot hold and then we set ${\mathcal{I}}_k=0$ and $p_k=0$.
The positive parameters $\phi_k$ in Algorithm~\ref{alg-bisection-inner} and $\phi_{P}$ in Algorithm~\ref{alg-bisection-two} are the thresholds determining the accuracy of the bisection search algorithms. Bisection search is a one-dimension search with very high efficiency and accuracy. For example, for an unit interval, after 30 iterations, the accuracy of a bisection search algorithm is $2^{-30}$. 
\begin{algorithm}[t]
\caption{The inner bisection search algorithm to find $p_k$ with given $\mu$.}
\label{alg-bisection-inner}
\begin{algorithmic}[1]
\STATE Initialize $\phi_k$, $p_{k,\min}$  and  $p_{k,\max}$;
\STATE Let $p_{k}=(p_{k,\min}+p_{k,\max})/2$;
\WHILE{$|f'_k(p_{k})- \mu|>\phi_{k}$}
\IF {$f'_k(p_{k})<\mu$}
\STATE $p_{k,\max}=(p_{k,\min}+p_{k,\max})/2$;
\ELSE
\STATE $p_{k,\min}=(p_{k,\min}+p_{k,\max})/2$;
\ENDIF
\STATE Let $p_k=(p_{k,\min}+p_{k,\max})/2$;
\ENDWHILE
\RETURN $p_{k}$.
\end{algorithmic}
\end{algorithm}
\begin{algorithm}[t]
\caption{The two dimensional bisection search algorithm to solve (\ref{water_c_1}) with given $\mathcal{I}_k$'s.}
\label{alg-bisection-two}
\begin{algorithmic}[1]
\STATE Initialize $\phi_{P}$, $\mu_{\min}$  and  $\mu_{\max}$;
\STATE Let $\mu=(\mu_{\min}+\mu_{\max})/2$;
\STATE Calculate $p_k$ if ${\mathcal{I}}_k=1$ by performing \textbf{Algorithm~\ref{alg-bisection-inner}} with $\mu$;
\WHILE{$|\sum_{k=1}^Kp_k- P|>\phi_{P}$}
\IF {$\sum_{k=1}^Kp_k<P$}
\STATE $\mu_{\max}=(\mu_{\min}+\mu_{\max})/2$;
\ELSE
\STATE $\mu_{\min}=(\mu_{\min}+\mu_{\max})/2$;
\ENDIF
\STATE Let $\mu=(\mu_{\min}+\mu_{\max})/2$;
\STATE Calculate $p_k$ if ${\mathcal{I}}_k=1$ via performing  \textbf{Algorithm~\ref{alg-bisection-inner}} with $\mu$;
\ENDWHILE
\RETURN $p_{k}$'s.
\end{algorithmic}
\end{algorithm}
Even under with the adoption of the bisection search, our proposed algorithm has fundamental difference to the ones based on solving KKT conditions and the gradient-based search algorithm. Specifically, solving KKT conditions involves the calculation of $K+1$ Lagrange multipliers as shown in (\ref{eq2}) and their computation may require $(K+1)$-dimensional search and well defined step sizes. The worst-case complexity is up to ${\mathcal{O}}(K^3)$.
But our proposed algorithm only requires two-dimensional bisection search which does not need the setup of step sizes (i.e., the concerned variable is automatically updated as the middle point in its considered region) and its worst-case complexity is much lower.

Interestingly, our proposed algorithm introduces the indicators $\mathcal{I}_k$'s in dealing with the nonnegative constraints and the partition of the subchannels into active ones and inactive ones (which can be further partitioned into upper-bound-tight-ones, lower-bound-tight-ones for more complicated problems). These ideas lead to efficient algorithms to find the solution values (an issue that was somewhat neglected in many papers on MIMO communication designs), and allow extensions to more complicated cases as shown later in this paper.

\subsection{Comparison and Application Examples}
\label{Sec-II-Ex}
The proposed new method, including the viewpoint and the algorithm, does not require manipulation of the KKT conditions into a format of  water-filling solutions. Further,  the proposed algorithm is general and has low-complexity with the worst-case number of iterations being $K-1$. On average, the number of iterations can be much smaller than $K-1$ since the proposed algorithm allows multiple channels to be made inactive in each iteration as long as their positivity constraints cannot be satisfied. For the traditional scheme, in general, all possible subsets of active subchannels need to be tested, whose complexity is exponential in $K$. For special cases when an ordering among the subchannel exists, the complexity can be reduced to $\mathcal{O}(K{\rm{log}}(K))$, which is still higher than the complexity of the proposed one. In what follows, examples are provided to better elaborate the difference and advantages of the proposed method.


\noindent \textbf{Example 1:} A general weighted sum capacity maximization problem has the following form:
\begin{align}
\label{Example_1}
\max_{p_1,\cdots,p_k} \ &{\sum}_{k=1}^Kw_k{\rm{log}}(b_k+a_kp_k) \nonumber \\
\ {\rm{s.t.}} \ & {\sum}_{k=1}^K p_k \le P, \ p_k \ge 0,
\end{align}where $w_k$, $b_k$ and $a_k$ are arbitrary non-negative parameters.

With our proposed scheme, we first obtain from the objective function in (\ref{Example_1})
\begin{equation}
g_k(\mu)={w_k}/{\mu}-{b_k}/{a_k}.
\label{ex1-1}
\end{equation}Then the solution values can be found by {\bf Algorithm \ref{algorithm_1}} within $K-1$ iterations. Specifically, from (\ref{water_c_1}),
\begin{equation}
\mu={{\sum}_{k=1}^K w_k \mathcal{I}_k }/\left({P+{\sum}_{k=1}^K\frac{b_k}{a_k}\mathcal{I}_k}\right).
\label{ex1-2}
\end{equation}The calculations in Step 2 and Step 6 can be achieved straightforwardly using (\ref{ex1-1}) and (\ref{ex1-2}).

With the traditional scheme, via calculations, the following water-filling solution is obtained:
\begin{align}
\label{Example_1_solution}
p_k=\left(\frac{w_k}{\mu}-\frac{b_k}{a_k}\right)^{+}, \quad {\sum}_{k=1}^K p_k = P.
\end{align}Though in compact neat form, to find the values of the optimal $p_k$'s is not self-explanatory. All possible active subchannel sets need to be tried to find the best one. In \cite{PeterHe2013}, an efficient water-filling algorithm was proposed for \textbf{Example 1}, where $a_iw_i/b_i$'s are in decreasing order. Our algorithm is essentially different from that in \cite{PeterHe2013} and the difference will be further elaborated at the end of this section.
%

\noindent \textbf{Example 2:} A general weighted MSE minimization problem can be written in the following form:
\begin{align}
\label{Example_2}
 \max_{p_1,\cdots,p_k} \ &{\sum}_{k=1}^K-\frac{w_k}{b_k+a_kp_k}  \nonumber \\
 {\rm{s.t.}} \ & {\sum}_{k=1}^K p_k \le P, \ p_k \ge 0,
\end{align}where $w_k$, $b_k$ and $a_k$ are arbitrary non-negative parameters.

With the proposed scheme, we first obtain from the problem
\begin{equation}
g_k(\mu)=\sqrt{\frac{w_k}{a_k\mu}}-\frac{b_k}{a_k}.
\label{ex2-1}
\end{equation}Similarly, {\bf Algorithm \ref{algorithm_1}} can be used to find the solution values. Specifically, from (\ref{water_c_1}),
\begin{equation}
\mu={{\sum}_{k=1}^K \sqrt{\frac{w_k}{a_k}}\mathcal{I}_k}/({P+{\sum}_{k=1}^K\frac{b_k}{a_k}\mathcal{I}_k}).
\label{ex2-2}
\end{equation}(\ref{ex2-1}) and (\ref{ex2-2}) can be used straightforwardly for the calculations in Steps 2 and 6.

With the traditional scheme, via calculations, the following water-filling solution is obtained as the first step:
\begin{align}
p_k=\left(\sqrt{\frac{w_k}{a_k\mu}}-\frac{b_k}{a_k}\right)^{+},\quad
{\sum}_{k=1}^K p_k = P.
\label{ex2-3}
\end{align}The same difficulty as in Example 1 appears here. Though (\ref{ex2-3}) is in compact neat form, it is unclear how to find the values of the optimal solution from it. In general all possible active subchannel sets need to be tried to find the best one whose complexity is exponential in $K$. Ordering of the subchannels is only possible with stringent ordering conditions on the parameters, e.g., $\sqrt{w_k/a_k}$ and $a_k/b_k$ can be ordered decreasingly simultaneously.

\noindent \textbf{Example 3:} The capacity maximization for dual-hop MIMO amplify-and-forward relaying networks can be casted as follows:
\begin{align}
\label{Example_3}
 \max_{p_1,\cdots,p_K} \ \ & {\sum}_{k=1}^K-w_k{\rm{log}}\left(1-\frac{a_kb_kp_k}{1+b_kp_k}\right) \nonumber \\
 \ {\rm{s.t.}} \ \ \ &{\sum}_{k=1}^K p_k \le P, \ \ p_k \ge 0,
\end{align}where $w_k,b_k$ are nonnegative and $0< a_k<1$.

With our proposed scheme, we can obtain from the objective function of the problem
\begin{align}
g_k(\mu)=\frac{\sqrt{a_k^2+
4w_k(1-a_k)a_kb_k/\mu}-(2-a_k)}{2(1-a_k)b_k}.
\end{align}
Then the solution values can be found by {\bf Algorithm \ref{algorithm_1}}. But for this case, to find the value of $\mu$ (for Steps 2 and 6), numerical bisection search is needed to solve the following equation
\begin{align}
\sum_{\{k|\mathcal{I}_k=1\}}\frac{\sqrt{a_k^2+
4w_k(1-a_k)a_kb_k/\mu}-(2-a_k)}{2(1-a_k)b_k}=P.
\end{align}

With the traditional scheme, via some calculations, the following water-filling solution is obtained as the first step:
\begin{align}
&p_k=\left(\frac{a_k-2+\sqrt{a_k^2+
4w_k(1-a_k)a_kb_k/\mu}}{2(1-a_k)b_k}\right)^{+}, \nonumber \\
&{\sum}_{k=1}^K p_k = P. 
\end{align}
But algorithms to find the water-filling solution values were not explicitly provided in existing literature.

\noindent \textbf{Example 4:} A weighted mutual information maximization problem for the training design can be written in the following format \cite{Bjornson2010}:
\begin{align}
\label{Example_4}
\max_{p_1,\cdots,p_K} \ \ & {\sum}_{j=1}^J {\sum}_{k=1}^K w_{k,j}{\rm{log}}(a_kc_j+b_kd_jp_k),\nonumber \\
 \ {\rm{s.t.}} \ \ & {\sum}_{k=1}^K p_k \le P, \ \ p_k \ge 0.
\end{align}

To use the proposed scheme, we first get from the objective function
\begin{equation}
f_k'(p_k)={\sum}_{j=1}^J \frac{w_{k,j}b_kd_j}{a_kc_j+b_kd_jp_k}.
\label{ex4-1}
\end{equation}
Due to the complexity of $f_k'(p_k)$, the inverse function $g_k(\mu)$ does not have an explicit analytical form. But since the derivative function $f_k'(p_k)$ is a decreasing function, its inverse function $g_k(\mu)$ and the sum function ${\sum}_{k=1}^Kg_k(\mu){\mathcal{I}}_k$ are also monotonically decreasing with respect to $\mu$. Given the analytical expression of the derivative in (\ref{ex4-1}) and the monotonically decreasing property, the values of  $\mu=f_k'(p_k)$ and its corresponding $p_k$ can be uniquely found through two dimensional bisection search. That is, the computations in Line 6 of Algorithm 1 can be done numerically and our proposed algorithm can still work.

%

\noindent \textbf{Example 5:} A  weighted MSE minimization problem for training optimization can be formulated as follows:
\begin{align}
\max_{p_1,\cdots,p_K} \ \ -&{\sum}_{j=1}^J {\sum}_{k=1}^K \frac{w_{k,j}}{a_kc_j+b_kd_jp_k}\nonumber \\
\ \ {\rm{s.t.}} \ \ &{\sum}_{k=1}^K p_k \le P, \ \ p_k \ge 0.
\end{align}

The derivative of the objective function is
\begin{align}
f_k'(p_k)={\sum}_{j=1}^J {w_{k,j}b_kd_j}/{(a_kc_j+b_kd_jp_k)^2}.
\end{align}
Similarly, the monotonic inverse function $g_k(\mu)$ cannot be written in an explicit analytical form, but {\bf Algorithm \ref{algorithm_1}} can still be used to find the solution values by calculating $\mu$ and $p_k$'s numerically in Steps 2 and 6. With the traditional scheme, similar to Example 4, the KKT conditions can be obtained but a compact water-filling solution form has not been found with the ``+'' operation and numerical searching algorithms are needed.


\subsection{Problems with Arbitrary Lower-Bound Constraints}
In this subsection, we consider the extension of the optimization problem P1 with arbitrary lower bounds on the subchannel powers:
\begin{align}
\label{opt_lower}
{\rm P1.1:} \  \max_{p_1,\cdots,p_K} \  &{\sum}_{k=1}^K f_k(p_k) \nonumber \\
{\rm{s.t.}} \ & {\sum}_{k=1}^K p_{k} \le P,  p_k\ge \gamma_k,
\end{align}where $P>0$ and $f_k(\cdot)$'s are real-valued, increasing, and strictly concave functions with continuous derivatives. In P1.1, each subchannel is limited with a non-negative lower bound for its power, while for P1, the lower bounds are zero for all subchannels. For this more general case, define the active set $\mathcal{S}_{active}$ as the set of subchannels whose powers are higher than their lower bounds, i.e.,
\begin{align}
\mathcal{S}_{active}\triangleq\{k|p_k>\gamma_k\}.
\end{align}The following lemma is obtained.
\begin{lemma}
\label{water_1}
The following conditions are both necessary and sufficient for the optimal solution of P1.1:
\begin{align}
\left\{\begin{array}{ll} f_k'(p_k)=f_j'(p_j) &  \text{for } k,j\in\mathcal{S}_{active}; \\
f_j'(p_j=\gamma_j)\le f_k'(p_k) & \text{for } k\in\mathcal{S}_{active} \text{ and}, j\notin\mathcal{S}_{active}; \\
\sum_{k=1}^K p_k=P.
\end{array} \right.
\label{eqn-P21}
\end{align}
\label{lemma-P21}
\end{lemma}\vspace{-3mm}

\begin{proof}
The proof is very similar to that of Lemma \ref{lemma-1}, thus omitted.
\end{proof}

For the algorithm design, the index operation ${\mathcal{I}}_k$ is introduced as follows: ${\mathcal{I}}_k=1$ when the power of subchannel $k$ is larger than its lower bound, i.e., $p_k>\gamma_k$; otherwise ${\mathcal{I}}_k=0$.
 Let $\mu=f_k'(p_k)$ for $k\in \mathcal{S}_{active}$, which  is the increasing rate for active subchannels. Via similar studies to those in Section \ref{sec-2-2},  the optimal solution of P1.1 can be represented as follows:
\begin{align}
\left\{\begin{array}{l} p_k=g_k(\mu){\mathcal{I}}_k+\gamma_k(1-{\mathcal{I}}_k) \\
P={\sum}_{k=1}^K\left[g_k(\mu){\mathcal{I}}_k
+\gamma_k(1-{\mathcal{I}}_k)\right]\end{array}\right..
\label{water_c_11}
\end{align}Notice that (\ref{water_c_1}) is a special case of (\ref{water_c_11}) where $\gamma_k=0$. {\bf Algorithm \ref{alg-2}} is proposed to find the solution values for P1.1.

\begin{algorithm}[t]
\caption{Proposed algorithm under arbitrary lower-bound constraints.}
\label{alg-2}
\begin{algorithmic}[1]
\STATE $\mathcal{I}_k=1$  $\text{for}$ $k=1,\cdots,K$;
\STATE $\text{Calculate}$ $\mu$ and $p_k$'s  using  Eqn.~(\ref{water_c_11}); \WHILE{$\text{length}( \text{find}(\{p_k\} <\{\gamma_k\} ))>0$}
\STATE Find $\mathcal{S}_{inactive}=\{k|p_k \le \gamma_k\}$;
\STATE $\text{Set}$ ${\mathcal{I}}_k=0$ and $p_k=\gamma_k$ $\text{for}$ $k\in \mathcal{S}_{inactive}$;
\STATE $\text{Calculate}$ $\mu$ and ${p_k}$'s using \text{Eqn.}~(\ref{water_c_11});
\ENDWHILE
\RETURN $p_k$'s.
\end{algorithmic}
\end{algorithm}

\begin{lemma}
\textbf{Algorithm~\ref{alg-2}} converges and achieves the optimal solution of P1.1.
\label{lemma-4}
\end{lemma}

\begin{proof}
The proof is similar to that of Lemma \ref{lemma-2}, thus omitted.
\end{proof}
In each iteration of \textbf{Algorithm~\ref{alg-2}}, subchannels whose powers are less than their required lower bounds are removed from the iteration (i.e., are put in the inactive set) and their powers are enforced to be the corresponding lower bounds, i.e., $p_k=\gamma_k$.  Since these subchannels are allocated smaller powers than their lower bounds, their increasing rates at the lower bounds $\gamma_k$ are smaller than other subchannels. After being removed, fewer power resources are available for the remaining active subchannels. After power allocation among the remaining subchannels in Step 6, the powers of the active subchannels decrease, and thus their increasing rates will increase. Therefore, the removed subchannels cannot enter the competition for power in future iterations. This explains the convergence and optimality of the algorithm intuitively. The worse case complexity order of Algorithm 2 is exactly the same as that of Algorithm 1, which is ${\mathcal{O}}(K^2)$.

\subsection{Discussions on More General Cases}
The new viewpoint and method can be extended to solve more general optimization problems in wireless communications. Consider the following convex optimization problem:
\begin{align}
\label{Opt_0}
{\rm{P2}}:\  \max_{p_1,\cdots,p_K} \ \ &{\sum}_{k=1}^K f_k(p_k) \nonumber \\
 \ {\rm{s.t.}} \ \ \ &{\sum}_{k=1}^K p_k \le P, \  h_l(p_k) \le0,
\end{align}where $P>0$ and $f_k(\cdot)$'s are real-valued, increasing, and strictly concave functions with continuous derivatives.

The difference of P2 to the original one P1 is in the constraints $h_l(p_k)$'s. When P2 is convex (e.g., when $h_l(p_k)\le 0$ can be transformed to a convex constraint on $p_k$), the following
KKT conditions are necessary and sufficient for the optimal solution \cite{Boyd04}:
\begin{align}
&f_k'(p_k)=\mu+{\sum}_l\lambda_lh_l'(p_k), \ \ \mu\left({\sum}_{k=1}^K p_k - P\right)=0, \nonumber \\
& \lambda_lh_l(p_k)=0, \ \ \mu\ge0, \ \ \lambda_k \ge 0,
\end{align}where $\mu$ and $\lambda_l$'s are the Lagrange multipliers corresponding to the sum power constraint and per-subchannel constraints, respectively.

By following the ideas proposed in previous subsections, we can solve this challenging problem by considering two situations: 1) all conditions $h_l(p_k)$'s are inactive (i.e., not satisfied with equality) and 2) at least one of  $h_l(p_k)$'s is active (i.e., satisfied with equality). The first situation leads to the same solution as P1. For the second one, the results for P1 can be applied for the power allocation among subchannels with inactive $h_l(p_k)$'s and solutions for subchannels with active $h_l(p_k)$'s can be found by solving $h_l(p_k)=0$.
In the following sections, we will solve the generalized problem considering several different cases.

\noindent \textbf{Remark:} The difference between our work and  \cite{PeterHe2013}
can be summarized as the difference between geometric and algebraic viewpoints. Each cannot include the other as its special case and each has its own advantages and characteristics.
Compared with the geometric logic, our logic has less geometric meanings. On the other hand, with the algebraic viewpoint,  our method can cover more mathematical formulations and tries to give a unified way for a broad range of water-filling solutions and water-filling algorithms.

\section{Problem with Box Constraints}
\label{sect_constrained_wf}
In this section, we consider a special case of P2 in which $L=1$ and $h_l(p_k)=(p_k-\gamma_k)(p_k-\tau_k)$. Equivalently, the optimization problem is as follows:\begin{align}
\label{Opt_2}
{\rm P3:} \ \max_{p_1,\cdots,p_K} \ & {\sum}_{k=1}^K f_k(p_k) \nonumber \\
{\rm{s.t.}} & {\sum}_{k=1}^K p_{k} \le P,  \gamma_k \le p_k \le \tau_k,
\end{align}where $P>0$, $\gamma_k\le\tau_k$, and $f_k(\cdot)$'s are real-valued, increasing, and strictly concave functions with continuous derivatives. The box constraint $\gamma_k \le p_k \le \tau_k$ is of practical importance \cite{PeterHe2013,PeterHe2017,PeterHeTSG2018,PeterHeTWC2018}. It is obvious that P3 is convex.


\subsection{Two Algorithms Built on Finding Subchannel Sets}
Similar to the previous section, we can  see this problem as dividing the available power $P$ into infinitesimally small portions $\delta_p$ and allocating them portion by portion. At the start of the allocation, Subchannel $k$ must have $\gamma_k$ to satisfy the lower bound constraint. For each remaining portion, we should choose the subchannel whose $f_k(\cdot)$ has the maximum increasing rate i.e., $f_k'(p_k)$, and whose power has not reached its upper bound to maximize the total of $f_k(\cdot)$'s. As the increasing rate of $f_k$ decreases when a power portion is added to it, after adding a portion to the subchannel with the maximum increasing rate, e.g., Subchannel $k$, its increasing rate may become smaller than the rate of another subchannel. In this case, a new subchannel $j$ with the fastest increasing rate will have the next power portion. Otherwise, Subchannel $k$ gets the next power portion if it still has the maximum $f_k'(p_k+\delta_p)$. This procedure repeats until all power portions have been allocated.
Some subchannels may never get any extra power portion than the original lower bounds when their increasing rates are never the highest. Some subchannels may have the highest increasing rates but cannot get more power due to their upper bound constraints. When the allocation stops, subchannels which do not have active bounds must have the same increasing rate.

For a given feasible solution $\{p_1,\cdots,p_K\}$, denote
\begin{align}
&\mathcal{S}_{l}\triangleq\{k|p_k=\gamma_k\}, \ \
S_{u}\triangleq\{k|p_k=\tau_k\}, \nonumber\\
&\mathcal{S}_{active}\triangleq\{k|\gamma_k< p_k<\tau_k\},
\label{S-III}
\end{align}which are the index sets of subchannels whose power values equal their lower bounds, upper bounds, and in-between the two bounds (i.e., active subchannels), respectively. They are also sets of subchannels with active lower bounds, active upper bounds, and no active bounds. The following lemma provides the sufficient and necessary condition on the optimal solution of P3. In our work, for the optimization problem P3 we mainly focus on the case that the sum power constraint is active as otherwise the optimization becomes very trivial. Specifically, when $\sum_{k=1}^K\tau_k < P$, the optimal solution is $p_k=\tau_k$. In this case, this is no need to design algorithms to solve P3. It is also worth highlighting that the following proposed algorithms can accommodate this trivial case directly.
\begin{lemma}
The following conditions are both necessary and sufficient for the optimal solution of P3:
\begin{align}
\left\{\begin{array}{ll} f_k'(p_k)=f_j'(p_j) & \text{for } k,j\in\mathcal{S}_{active}; \\
f_j'(p_j=\gamma_j)\le f_k'(p_k) & \text{for } k\in\mathcal{S}_{active} \text{ and } j\in\mathcal{S}_{l}; \\
f_j'(p_j=\tau_j)\ge f_k'(p_k) & \text{for } k\in\mathcal{S}_{active} \text{ and } j\in\mathcal{S}_{u}; \\
\sum_{k=1}^K p_k=\widetilde{P} & \widetilde{P}=\min\{P,\sum_{k=1}^K\tau_k\}.
\end{array} \right.
\label{cond-2}
\end{align}
\label{lemma-5}
\end{lemma}
\begin{proof}
The proof is similar to that of Lemma \ref{lemma-1} with the following two changes: 1) the lower bounds change from 0 to $\gamma_k$'s and 2) new upper bounds are added. Details are omitted to save space.
\end{proof}

The physical meaning of (\ref{cond-2}) is as follows. At the optimal solution, subchannels with inactive bounds have the same increasing rate $f_k'(p_k)$, which is also denoted as $\mu$. Subchannels with active lower bounds have lower increasing rates than $\mu$ and subchannels with active upper bounds have higher increasing rates than $\mu$.

\begin{algorithm}[t]
\caption{The first proposed algorithm for P3.}
\label{alg-P3-a}
\begin{algorithmic}[1]
\STATE $\text{Perform}$ $\textbf{Algorithm~\ref{alg-2}}$;
\WHILE{$\text{length}( \text{find}(\{p_k\} > \{\tau_k\} ))>0$}
\STATE Find $\mathcal{S}_u=\{k|p_k \ge \tau_k\}$ and $\text{set}$ $p_k=\tau_k$ $\text{for}$ $k\in \mathcal{S}_u$;
\STATE Find $\mathcal{S}_{other}=\{k|p_k < \tau_k\}$;

\STATE  $\text{Update}$ $P\leftarrow P-\sum_{k\in \mathcal{S}_u}\tau_k$;
\STATE  $\text{Perform}$ $\textbf{Algorithm~\ref{alg-2}}$ for subchannels in  $\mathcal{S}_{other}$ with the updated total power $P$;
\ENDWHILE
\RETURN $p_k$'s.
\end{algorithmic}
\end{algorithm}

Based on the viewpoint and conditions for the optimal solution of P3, we propose {\bf Algorithm \ref{alg-P3-a}} to find the solution values by using Algorithm \ref{alg-2} as a building block. The idea is to first consider the lower bound constraints only and use Algorithm \ref{alg-2} to find the corresponding solution. Then the subchannels whose power values are larger or the same as their upper bounds are re-set as their upper bounds, and are removed from the set of active subchannels. In the next iteration, power is allocated among the remaining active subchannels using Algorithm \ref{alg-2} again. The process continues until the powers of all subchanels are smaller than their upper bounds at an iteration. Algorithm \ref{alg-P3-a} has one more round of iteration than Algorithm \ref{alg-2}. Thus its worse case complexity order is $\mathcal{O}(K^3)$. Algorithm \ref{alg-P3-a} does not have balanced treatment between the lower bound constraints and the upper bound constraints. While subchannels that reach or violate their upper bound constraints are removed during the iterations, the ones reaching or violating their lower constraints stay in the `while' loop and participate in the power allocation procedure with Algorithm \ref{alg-2}. Another algorithm symmetrical to Algorithm \ref{alg-P3-a} can also be designed by switching the roles of the lower and upper bound constraints.

\begin{algorithm}[t]
\caption{The second proposed balanced algorithm for P3.}
\label{alg-P3-b}
\begin{algorithmic}[1]
\STATE Initialize $\mathcal{I}_k=1$ and $\mathcal{J}_k=1$ $\text{for}$ $k=1,\cdots,K$;
\STATE $\text{Calculate}$ $\mu$ and $p_k$'s $ \text{using} $  Eqn.~(\ref{water_c_P3-32});
\WHILE{$\text{length}(\text{find}(\{p_k\}\! <\! \{\gamma_k\} ))\!+\!\text{length}( \text{find}(\{p_k\} \!>\! \{\tau_k\} ))\!\! >\!\! 0$}
\STATE Find $\mathcal{S}_{l}=\{k|p_k\le\gamma_k\}$ and set $p_k=\gamma_k$, $\mathcal{I}_k=0$ for $k\in \mathcal{S}_l$;
\STATE  $\text{Calculate}$ $\mu$ and ${p_k}$'s using Eqn.~(\ref{water_c_P3-32});
\IF{$\text{length} (\text{find}(\{p_k\}\!<\!\{\gamma_k\}))=0$ $\&$  $\text{length}( \text{find}(\{p_k\} \!>\! \{\tau_k\} ))>0$}
\STATE Find $\mathcal{S}_{u}=\{k|p_k\ge\tau_k\}$ and $\text{set}$ $p_k=\tau_k$, $\mathcal{J}_k=0$ for $k\in \mathcal{S}_u$;
\STATE Set $\mathcal{I}_k=1$  $\text{for}$ $k=1,\cdots,K$;
\STATE  $\text{Calculate}$ $\mu$ and ${p_k}$'s using Eqn.~(\ref{water_c_P3-32});
\ENDIF
\ENDWHILE
\RETURN $p_k$'s
\end{algorithmic}
\end{algorithm}Next, we consider both constraints jointly. Based on the aforementioned discussions, the key task is to determine the sets $\mathcal{S}_u$, $\mathcal{S}_l$, and $\mathcal{S}_{active}$ defined in (\ref{S-III}). We  introduce two sets of indices $\mathcal{I}_k$'s and $\mathcal{J}_k$'s as follows:
\begin{align}
\mathcal{I}_k=\left\{\begin{array}{ll}1 & \text{ if }  p_k>\gamma_k \\
0 & \text{otherwise}
\end{array}\right.,\
\mathcal{J}_k=\left\{\begin{array}{ll} 1 & \text{ if }  p_k<\tau_k \\
0 & \text{otherwise}
\end{array}\right.,
\end{align}where $p_k$ is the power allocated to subchannel $k$,
$\mathcal{I}_k$ indicates whether the power of subchannel $k$ is larger than its lower bound constraint and $\mathcal{J}_k$ indicates whether it is smaller than its upper bound constraint. For the index tuple $(\mathcal{I}_k,\mathcal{J}_k)$, $(1,1)$ means the subchannel is an active one and neither constraints is tight; $(1,0)$ means the subchannel belongs to $\mathcal{S}_u$; and $(0,1)$ means the subchannel belongs to $\mathcal{S}_l$. 
Similar to the previous section, let $\mu$ be the increasing rate of the active subchannels, and we have the following necessary conditions for P3 from (\ref{cond-2}): \begin{align}
\hspace{-3mm}\left\{\hspace{-1mm}\begin{array}{l} p_k=g_k(\mu)\mathcal{I}_k\mathcal{J}_k+\gamma_k(1-{\mathcal{I}}_k)+\tau_k(1-{\mathcal{J}}_k), \\
P={\sum}_{k=1}^K\left[g_k(\mu){\mathcal{I}}_k\mathcal{J}_k+\gamma_k(1-{\mathcal{I}}_k)+\tau_k(1-{\mathcal{J}}_k)
\right]\end{array}\right.\hspace{-2mm}.
\label{water_c_P3-32}
\end{align}Algorithm \ref{alg-P3-b} is proposed which follows the idea in Algorithm \ref{alg-2} with extensions for both lower and upper bound constraints. The worse case complexity order of Algorithm \ref{alg-P3-b} is the same as that of Algorithm \ref{alg-2}, i.e., ${\mathcal{O}}({K^2})$. By using results in Lemma \ref{lemma-5} and following the proof in Lemma \ref{lemma-2}, the convergence and optimality of Algorithms \ref{alg-P3-a} and \ref{alg-P3-b} can be proved.
\begin{lemma}
\textbf{Algorithms~\ref{alg-P3-a}} and \textbf{\ref{alg-P3-b}} converge and achieve the optimal solution of P3.
\label{lemma-6}
\end{lemma}
\begin{proof}
The detailed proof is similar to that of Lemma \ref{lemma-2} and is thus omitted to save space.
\end{proof}

\subsection{Two Algorithms Built on Finding the Final Increasing Rate}

\begin{algorithm}[t]
\caption{The third proposed bisection search based algorithm for P3.}
\label{alg-P3-c}
\begin{algorithmic}[1]
\STATE Initialize $\phi_{P}$, $\mu_{\min}$, $\mu_{\max}$ and  $\mathcal{I}_k=\mathcal{J}_k=1$ ;
\IF {$\sum_{k=1}^K\tau_k<P$}
\STATE Let $\mu=(\mu_{\min}+\mu_{\max})/2$ and calculate $p_k$'s using the first formula of Eqn.~(\ref{water_c_P3-32});
\STATE Find $\mathcal{S}_u=\{k|p_k > \tau_k\}$ and $\text{set}$ $p_k=\tau_k$ $\text{for}$ $k\in \mathcal{S}_u$;
\STATE Find $\mathcal{S}_l=\{k|p_k < \tau_k\}$ and $\text{set}$ $p_k=\gamma_k$ $\text{for}$ $k\in \mathcal{S}_l$;
\WHILE{$|\sum_{k=1}^Kp_k- P|>\phi_{P}$ }
\IF {$\sum_{k=1}^Kp_k<P$}
\STATE $\mu_{\max}=(\mu_{\min}+\mu_{\max})/2$;
\ELSE
\STATE $\mu_{\min}=(\mu_{\min}+\mu_{\max})/2$;
\ENDIF
\STATE Let $\mu=(\mu_{\min}+\mu_{\max})/2$ and calculate $p_k$'s  using the first formula of Eqn.~(\ref{water_c_P3-32}).
\STATE Find $\mathcal{S}_u=\{k|p_k > \tau_k\}$ and $\text{set}$ $p_k=\tau_k$ $\text{for}$ $k\in \mathcal{S}_u$;
\STATE Find $\mathcal{S}_l=\{k|p_k < \tau_k\}$ and $\text{set}$ $p_k=\gamma_k$ $\text{for}$ $k\in \mathcal{S}_l$;
\ENDWHILE
\ELSE
\STATE Set $p_k=\tau_k$ for $k=1,\cdots,K$;
\ENDIF
\RETURN $p_{k}$'s.
\end{algorithmic}
\end{algorithm}

In this subsection, two new algorithms are proposed, which are constructed by finding the final increasing rate $\mu$ of active subchannels. We first propose a complex but general one in \textbf{Algorithms~\ref{alg-P3-c}}, where bisection search is used.  It is easy to understand and implement, but suffers high complexity and numerical accuracy limitations. The complexity of Algorithm~\ref{alg-P3-c} is still ${\mathcal{O}}(K^2)$.\begin{algorithm}[t]
\caption{The fourth proposed order-based algorithm for P3.}
\label{alg-P3-d}
\begin{algorithmic}[1]
\STATE $\text{Order}$ $f_k'({\tau_k})$'s decreasingly, i.e., via (\ref{order-1}), to obtain the index sequence $\{\sigma_1,\cdots\sigma_K\}$;
\STATE Set $i=1$ and $\tau_{\sigma_0}=0$;
\STATE Calculate  $\mu= f_{\sigma_i}'({\tau_{\sigma_i}})$;
\STATE Calculate $p_k$ from (\ref{power-Alg-P3-d}) with $\mu$;
\WHILE{$\sum_{k=1}^K p_k<P$ $\&$ $i< K$}
\STATE Update $i\leftarrow i+1$;
\STATE Calculate $\mu= f_{\sigma_i}'({\tau_{\sigma_i}})$;
\STATE Calculate $p_k$ $\text{from}$ (\ref{power-Alg-P3-d}) with $\mu$;

\ENDWHILE
\IF {$\sum_{k=1}^K p_k>P$}
\STATE Perform \textbf{Algorithm \ref{alg-2}} on subchannels $\{\sigma_k| i \le k \le K \}$ with the total power being $P-\sum_{n=0}^{i-1}\tau_{\sigma_n}$.
\ENDIF
\RETURN $p_k$'s.
\end{algorithmic}
\end{algorithm} The next algorithm, \textbf{Algorithm \ref{alg-P3-d}}, uses a more efficient method to find the final increasing rate $\mu$ whose complexity order is also ${\mathcal{O}}(K^2)$. First, the subchannels are ordered decreasingly based on their increasing rates at the power upper bounds $f_k'(\tau_k)$'s such that
\begin{equation}
f_{\sigma_1}'(\tau_{\sigma_1})\ge \cdots\ge f_{\sigma_K}'(\tau_{\sigma_K}).
\label{order-1}
\end{equation}From the results in Lemma \ref{lemma-5}, it can be shown that the subchannel with a higher $f_k'(\tau_k)$ has higher priority to achieve its upper bound. In other words, if at the optimal solution $p_{\sigma_i}=\tau_{\sigma_i}$, then $p_{\sigma_j}=\tau_{\sigma_j}$ for all $j<i$. Thus, in finding the optimal solution, we can consider the cases of $p_{\sigma_{[1:i]}}=\tau_{\sigma_{[1:i]}}$ and $p_{\sigma_{[i+1:K]}}<\tau_{\sigma_{[i+1:K]}}$ for $i=1,\cdots,K-1$, sequentially starting with $i=1$. That is, the $i$th case corresponds to $\mathcal{S}_u=\{\sigma_1,\sigma_2,\cdots,\sigma_i\}$. Notice that the $i$th case happens if and only if $\mu\in(f_{\sigma_i}'(\tau_{\sigma_i}), f_{\sigma_{i-1}}'(\tau_{\sigma_{i-1}})]$, where we define $f_{\sigma_0}'(\tau_{\sigma_0})=\infty$. Thus this is equivalent to considering that $\mu$ is in the intervals $(f_{\sigma_i}'(\tau_{\sigma_i}), f_{\sigma_{i-1}}'(\tau_{\sigma_{i-1}})]$ for $i=1,2,\cdots,K$, sequentially to decide the correct $\mu$ interval.

In dealing with the $i$th case, let $\mu=f_{\sigma_i}'(\tau_{\sigma_i})$, and based on Lemma 3 the power for each subchannel is given by
\begin{equation}
\left\{\begin{array}{ll}
p_{\sigma_k}=\tau_{\sigma_k} & \text{if} \ k\le i, \\
p_{\sigma_k}=g_{\sigma_k}(\mu) &  \text{if} \ k>i,f_{\sigma_k}'(\gamma_k)> \mu \\
P_{\sigma_k}=\gamma_k & \text{if} \ k>i, f_{\sigma_k}'(\gamma_k)\le \mu.
\end{array}\right.
\label{power-Alg-P3-d}
\end{equation}Then the total power is calculated and compared with the power constraint $P$. If $\sum_{k=1}^Kp_k> P$, none of the subchannels $\sigma_{i},\cdots,\sigma_{K}$ can reach its upper bound. Thus the optimal solution of P3 falls into this case. As for this case, subchannels $\sigma_{i},\cdots,\sigma_{K}$ have inactive upper bounds, the bounds can be ignored and \textbf{Algorithm~\ref{alg-2}} can be used to find the optimal values of their powers. If $\sum_{k=1}^Kp_k<P$, the increasing rate $\mu=f_{\sigma_i}'(\tau_{\sigma_i})$ is too high for all $p_{\sigma_{i+1}},\cdots,p_{\sigma_{K}}$ to stay below their upper bounds.
As a result, Case $i$ is not the optimal and the next case should be considered. Note that when $\sum_{k=1}^Kp_k=P$, Case $i$ is optimal. If the last case, Case $K$ is considered, and still $\sum_{k=1}^Kp_k< P$, this means $P>\sum_{i=1}^K\tau_k$ and all subchannels should use their maximum powers. With the above discussions and Lemma \ref{lemma-5}, the following lemma can be proved.
\begin{lemma}
\textbf{Algorithms \ref{alg-P3-c}} and \textbf{\ref{alg-P3-d}} converge and achieve the optimal solution of P3.
\end{lemma}

\subsection{Application Examples}
In this subsection, a few application examples are given.

\noindent \textbf{Example 6:} A weighted capacity maximization problem under box constraints can be formulated as follows:
\begin{align}
 \max_{p_1,\cdots,p_K} \  & {\sum}_{k=1}^Kw_k{\rm{log}}(b_k+a_kp_k) \nonumber \\
 {\rm{s.t.}} \  &{\sum}_{k=1}^K p_k \le P, \ \gamma_k \le p_k \le \tau_k.
\end{align}

\noindent \textbf{Example 7:} A weighted MSE minimization problem for MIMO-OFDM systems under box constraints is formulated as
\begin{align}
 \max_{\{p_{j,k}\}} \ & {\sum}_{j=1}^J{\sum}_{k=1}^K-{w_{j,k}}/({b_{j,k}+a_{j,k}p_{j,k}}) \nonumber \\
 \ {\rm{s.t.}} \ &{\sum}_{j=1}^J{\sum}_{k=1}^K p_{j,k} \le P, \ \gamma_{j,k} \le p_{j,k}\le \tau_{j,k}.
\end{align}

\noindent \textbf{Example 8:} The weighted capacity maximization problem for AF MIMO relaying systems can be written in the following form\begin{align}
\max_{p_1,\cdots,p_K} \  & {\sum}_{k=1}^K-w_k{\rm{log}}\left(1-{a_kb_kp_k}/({1+b_kp_k})\right) \nonumber \\
  {\rm{s.t.}} \  &{\sum}_{k=1}^K p_k \le P,  \ \gamma_k \le p_k\le \tau_k.
\end{align}

For the problems in Examples 6-8, the proposed algorithms in Algorithms \ref{alg-P3-a}-\ref{alg-P3-d} can be used to find the optimal solutions.
%

\section{Several Extensions}
\label{sec-extension}
In this section, the proposed viewpoint and algorithms are extended to several more complicated optimization problems. 

\subsection{Problems with Multiple Ascending Sum-Constraints}

We first investigate the extension of P3 to have multiple ascending sum-constraints \cite{Antonio2014}. The optimization problem is posed as follows:
\begin{align}
{\rm P4}:\!\!\! \max_{p_1,\cdots,p_K} \  & {\sum}_{k=1}^K f_k(p_k) \nonumber \\
 \ {\rm{s.t.}} \  & {\sum}_{k=1}^J p_{k} \le P_J, \ J=1,\cdots,K \nonumber \\
 \  & \gamma_k \le p_k \le \tau_k, \ \ k=1,\cdots,K,
\label{Cond-41}
\end{align}where $f_k$'s have the same properties as in P3. The main difference to P3 is that P4 has a total of $K$ constraints on the ascending accumulative sums, while P3 has one total sum-constraint. Thus the feasible region of P3 is no smaller than that of P4. A new approach different from that in \cite{Antonio2014} is proposed here, where the algorithms proposed for P3 in the previous section are used as building blocks.

Without loss of generality, for the first set of constraints in (\ref{Cond-41}), the index set of active constraints is denoted as $\{J_{1}^*,\cdots J_{N}^*\}$, i.e., $\sum_{k=1}^{J_n^{*}}p_{k}=P_{J_{n}^*}$ for $1,\cdots,N$ while $\sum_{k=1}^J p_k<P_n$ for $J\ne J_1^*,\cdots, J_N^*$. For completeness, we define $J_0^*=1$  $P_{0}=0$ and further define the subchannel set ${\mathcal C}_{n}=\{k|J_{n-1}^* < k  \le J_{n}^* \}$ for $n=1,\cdots,N$. For each ${\mathcal C}_{n}$, three subsets are defined, i.e.,  $S_{n,active}=\{k|\gamma_k<p_k<\tau_k,k\in{\mathcal C}_{n}\}$, $S_{n,l}=\{k|p_k=\gamma_k,k\in {\mathcal C}_{n}\}$ and $S_{n,u}=\{k|p_k=\tau_k,k\in {\mathcal C}_{n}\}$. It can be seen that ${\mathcal C}_{1},\cdots,{\mathcal C}_{N}$ form a partition of $\{1,\cdots, J_N^*\}$ and $S_{n,active},S_{n,l},S_{n,u}$ form a partition of $\mathcal{C}_{n}$.  Based on these definitions and Lemma 5, we have the necessary and sufficient conditions for the solution of P4 in the following lemma.
\begin{algorithm}[t]
\caption{Proposed recursive nested algorithm for P4.}
\label{alg-P4}
\begin{algorithmic}[1]
\STATE Initialize $P_0=0$ and $N_{\min}=1$; 
\WHILE{$N_{\min} < K$}
\STATE Use one of $\textbf{Algorithms~\ref{alg-P3-a}-\ref{alg-P3-d}}$ to calculate $p_{k,J}$'s over $\{N_{\min},\cdots,J\}$ for $J=N_{\min},\cdots,K$ under the sum-constraint ${\sum}_{k=N_{\min}}^J p_{k} \le P_J-P_{N_{\min}-1}$ and the box constraints;
\STATE Calculate the increasing rates $f'^{(J)}$'s of the active subchannels for $J=N_{\min},\cdots,K$ and set $f'^{(J)}=0$ when there is no active subchannel for $J$;
\STATE Find $J_{\min}$ for which the increasing rate $f'^{(J_{\min})}$ has the maximum value; and if there are multiple values having the same maximum increasing rate, $J_{\min}$ is the maximum value;
\STATE Set $p_{k}=p_{k,J_{\min}}$ for $k=N_{\min},\cdots,J_{\min}$;
\STATE Set $N_{\min}=J_{\min}+1$;
\ENDWHILE
\RETURN $p_k$'s.
\end{algorithmic}
\end{algorithm}
\begin{lemma}
The following conditions are both necessary and sufficient for the optimal solution of P5:
\begin{align}
\label{Con-P5}
 \hspace{-2mm}\left\{\begin{array}{ll} f_k'(p_k)=f_j'(p_j), \ \text{for } k,j\in {\mathcal{S}}_{n,active}; \\
f_j'(p_j=\gamma_j)\le f_k'(p_k), \ \text{for } k\in\mathcal{S}_{n,active}, j\in\mathcal{S}_{n,l}; \\
f_j'(p_j=\tau_j)\ge f_k'(p_k),\  \text{for } k\in\mathcal{S}_{n,active}, j\in\mathcal{S}_{n,u}; \\
\sum_{k=J_{n-1}^*+1}^{J_{n}^*}p_{k}=P_{J_{n}^*}\hspace{-1mm} -\hspace{-1mm} P_{J_{n-1}^*}, \ \text{for}\ n=1,\cdots,N\\
\sum_{k=1}^{J}p_{k}=\min\{P_J,P_{J_{N}^*}+\sum_{k=J_N^*+1}^J\tau_{k}\}, J_N^*<J\le K,  \\
f_{k_1}'(p_{k_1}) \ge f_{k_2}'(p_{k_2}) \ge \hspace{-1mm}\cdots \hspace{-1mm} \ge f_{k_N}'(p_{k_N}), \text{for} \ k_n\in \mathcal{S}_{n,active}.
\end{array} \right.
\end{align}
\label{Lemma-8}
\end{lemma}

\begin{proof}
It can be seen that the power optimization for the subchannel set ${\mathcal C}_{n}$ in  P4 becomes P3 with $P=P_{J_{n}^*}-P_{J_{n-1}^*}$ as otherwise based on the discussions for Lemma 5, the sum performance for subchannels in ${\mathcal C}_{n}$ can be further improved without violating the constraints in P4. In other words, the optimal solutions of $p_k$ in each ${\mathcal C}_{n}$ are given by Lemma 5 with $P=P_{J_{n}^*}-P_{J_{n-1}^*}$. This shows the necessity of the conditions in the first four lines of (\ref{Con-P5}) for the optimal solution. The necessity of the second last condition in (\ref{Con-P5}) can be seen as follows. From the definition of $J_N^*$, for $J=J_N^*+1$, the sum-power constraint is not tight meaning $\sum_{k=1}^{J_N^*+1} p_k<P_{J_N^*+1}$ while $\sum_{k=1}^{J_N^*} p_k=P_{J_N^*}$. This happens only when $P_{J_{N}^*+1} > P_{J_N^*}+\tau_{J_N^*+1}$. Thus the optimal solution for $p_{J_N^*+1}$ is to take its upper bound $\tau_{J_N^*+1}$. The same goes to all $J_N^*< J \le K$.

Finally, for the last condition in (\ref{Con-P5}), we prove the necessity of the first part: $f_{k_1}'(p_{k_1})> f_{k_2}'(p_{k_2})$ for $k_1\in \mathcal{S}_{1,active},k_2\in \mathcal{S}_{2,active}$, via contradiction. The remaining part of the inequality can be proved similarly. Assume that $p_1^*,\cdots,p_K^*$ is the optimal solution but $f_{k_1}'(p^*_{k_1})< f_{k_2}'(p^*_{k_2})$ for $k_1\in \mathcal{S}_{1,active},k_2\in \mathcal{S}_{2,active}$. Since $f_{k_1}'$ and $f_{k_2}'$ are continuous, there exists a $\delta$ with $0<p_{k_1}^*-\gamma_{k_1}<\delta<\tau_{k_2}-p_{k_2}^*$ such that $f_{k_1}'(p^*_{k_1}-x)<f_{k_2}'(p^*_{k_2}+x)$ for $0< x\le \delta$. Thus
\begin{align}
&f_{k_1}(p^*_{k_1}-\delta)+f_{k_2}(p^*_{k_2}+\delta) \nonumber  \\
&=f_{k_1}(p^*_{k_1})+f_{k_1}(p^*_{k_1})+\hspace{-2mm}\int_{0}^{\delta}\hspace{-2mm}
\left[f_{k_2}'(p^*_{k_2}+x)-f_{k_1}'(p^*_{k_1}-x)\right]dx \nonumber \\
&>f_1(p^*_1)+f_2(p^*_2).
\end{align}This shows that the new solution where $p^*_{k_1},p^*_{k_2}$ are replaced by $p^*_{k_1}-\delta,p^*_{k_2}+\delta$ with the remaining $p_{k}$'s unchanged achieves a high objective value. With the conditions on $\delta$, it can be easily shown that the new solution is feasible, i.e., satisfies all box constraints and sum constraints. This contradicts with the assumption that $p_1^*,\cdots,p_K^*$ is the optimal solution.

For the sufficiency, since the power optimization for the subchannel set ${\mathcal C}_{n}$ in P4 becomes P3 with $P=P_{J_{n}^*}-P_{J_{n-1}^*}$, with the sufficiency proved in Lemma \ref{lemma-5}, it suffices to show that any two vectors, $\{p_{1,1},\cdots,p_{K,1}\}$ and $\{p_{1,2},\cdots,p_{K,2}\}$, satisfying (\ref{Con-P5}) must have the same set of active sum-constraints, i.e., $\{J_{1,1}^*,\cdots,J_{N_1,1}^*\}=\{J_{1,2}^*,\cdots,J_{N_2,2}^*\}$. In what follows, we prove this by contradiction from the last index sequentially to the fist index. Assume that $J_{N_1,1}^*\neq J_{N_2,2}^*$ and without loss of generality, assume that $J_{N_1,1}^*> J_{N_2,2}^*$. We have the following conclusions from the definition of $J_{N_1}^*,J_{N_2}^*$:
\begin{align*}
& {\sum}_{k=1}^{J_{N_2}^*}p_{k,1}\le P_{J_{N_2}^*}={\sum}_{k=1}^{J_{N_2}^*}p_{k,2},\\
& {\sum}_{k=1}^{J_{N_1}^*}p_{k,1}= P_{J_{N_1}^*}>{\sum}_{k=1}^{J_{N_1}^*}p_{k,2}.
\end{align*}
Thus $\sum_{k=J_{N_1}^*+1}^{J_{N_2}^*}p_{k,1}>\sum_{k=J_{N_1}^*+1}^{J_{N_2}^*}p_{k,2}=\sum_{k=J_{N_1}^*+1}^{J_{N_2}^*}\tau_k$,
where the last step is because of the condition on the fifth line of (\ref{Con-P5}). This is impossible since $p_{k,1}\le \tau_{k}$ for all $k$. Thus we have $J_{N_1,1}^*= J_{N_2,2}^*=J_N^*$. Now we move on to show $J_{N_1-1,1}^* = J_{N_2-1,2}^*$ via contradiction. Assume without loss of genitality that $J_{N_1-1,1}^* > J_{N_2-1,2}^*$. To help the presentation, define $X_1=\{n| J_{N_2-1,2}^*< n \le J_{N_1-1,1}^*\}$ and $X_2=\{n|J_{N_2-1,1}^*< n \le J_{N}^*\}$. Moreover $V_1= \hspace{-1mm}\sum_{n\in X_1}p_{n,2}$, $V_2= \hspace{-1mm}\sum_{n\in X_2}p_{n,2}$, $V_3= \hspace{-1mm}\sum_{n\in X_1}p_{n,1}$ and
$V_4= \hspace{-1mm}\sum_{n \in X_2 }p_{n,1}$. Similarly, from the definition of $J_{N_1-1,1}^*, J_{N_2-1,2}^*, J_N^*$, we have $V_2>V_4$ and $V_3+V_4>V_1+V_2$, which leads to $V_3>V_1$. Consider the power optimization for the subchannel set $X_1$ with sum-constraint $V_3$ and $V_1$ respectively. We have $f_{m}'(p_{m,1}) < f_{k}'(p_{k,2})$ for active subchannels $m,k \in X_1$ since $V_3>V_1$. On the other hand, as $V_2>V_4$,  it can be concluded that $f_{j}'(p_{j,1}) > f_{k}'(p_{k,2})$ for active subchannels $j,k \in X_2$. From the first condition of (\ref{Con-P5}), for the second solution, active channels in both $X_1$ and $X_2$ have the same $f_k'(p_{k,2})$. Thus, $f_{m}'(p_{m,1}) < f_{k}'(p_{j,1})$ for active channels $m\in X_1,j\in X_2$. This contradicts the last condition in (\ref{Con-P5}). The only possibility remaining is that the active subchannel set for solution 1 in $X_2$ to empty and the active subchannel set for solution 2 in $X_1\cup X_2$ is empty, and thus $p_{k,1}=\tau_{k}$ for $k>J_{N_1^*,1}+1$ and $p_{k,2}=\tau_{k}$ for $k>J_{N_2^*,2}+1$. In this situation, the sets of active sum-constraints for the two solutions reduces to $\{J_{1,1}^*,\cdots,J_{N_1-1,1}^*\}$ and $\{J_{1,2}^*,\cdots,J_{N_2-1,2}^*\}$, where the last elements can be omitted. By using the previous argument, we have $J_{N_1-1,1}^*=J_{N_2-1,2}^*$. The process can be repeated to show $\{J_{1,1}^*,\cdots,J_{N_1,1}^*\}=\{J_{1,2}^*,\cdots,J_{N_2,2}^*\}$.
\end{proof}

\begin{lemma}
Algorithm \ref{alg-P4} converges and achieves the optimal solution of P5.
\end{lemma}

\begin{proof} It can be seen from Step 7 of the algorithm that $N_{\min}$ increases for each iteration round of the `while' loop. Thus the algorithm will terminate after at most $K-1$ rounds. Since Algorithms 5-8 converge, Algorithm \ref{alg-P4} converges. We will show that the solution Algorithm \ref{alg-P4} finds satisfies the conditions in Lemma \ref{Lemma-8}, which are both necessary and sufficient for the optimal solution. It can be seen from Step 3 and Step 6 that Algorithm~\ref{alg-P4} is built on Algorithms~\ref{alg-P3-a}-\ref{alg-P3-d}. Thus it can directly be concluded from Lemma \ref{lemma-5} and Lemma \ref{lemma-6} that the first five lines of (\ref{Con-P5}) hold. In the following, we show that the solution also satisfies the final line in (\ref{Con-P5}). By the iterative nature of the algorithm, we only need to show the first part: $f'_{k_1}(p_{k_1})\ge f'_{k_2}(p_{k_2})$ for $k_1\in\mathcal{S}_{1,active}$ and $k_2\in\mathcal{S}_{2,active}$.

Denote the $J_{\min}$ value found in Step 5 in the $1$st iteration of the `while' loop as $J_1$, the corresponding maximum increasing rate as $f'^{(J_1)}$ (which is the same as $f'_{k_1}(p_{k_1})$ for $k_1\in\mathcal{S}_{1,active}$), and the corresponding active, lower-bound achieving, upper-bound achieving subchannel sets as $\mathcal{S}_{1,active},\mathcal{S}_{1,l},\mathcal{S}_{1,u}$. Denote the $J_{\min}$ value found in Step 5 in the $2$nd iteration of the `while' loop as $J_2$, the corresponding maximum increasing rate as $f'^{(J_2)}$ (which is the same as $f'_{k_2}(p_{k_2})$ for $k_2\in\mathcal{S}_{2,active}$), and the corresponding subchannel sets as $\mathcal{S}_{2,active},\mathcal{S}_{2,l},\mathcal{S}_{2,low}$.
Further denote the solution found in Step 3 of \textit{the first iteration} for $J=J_2$ as $p_{1,J_2},p_{2,J_2},\cdots,p_{J_2,J_2}$, the corresponding increasing rate of active subchannels as $f'^{(J_2,old)}$ and the corresponding subchannel sets as $\mathcal{S}_{J_2,active,old},\mathcal{S}_{J_2,l,old},\mathcal{S}_{J_2,u,old}$. We need to show that $f'^{(J_1)}\ge f'^{(J_2)}$.

Case 1: $\mathcal{S}_{1,active}=\emptyset$. In this case, from Step 4, we have  $f'^{(J_1)}=0$. As it is the maximum increasing rate, the increasing rates found in Step 4 for all $J=1,\cdots,K$ are zero. Thus $J_{1}=K$ via Step 5, and the algorithm terminates. The optimization problem reduces to the scenario that none of the sum-constraint is active and our algorithm finds the optimal solution.

Case 2: $\mathcal{S}_{1,active}\ne \emptyset$, equivalently $f'^{(J_1)}>0$. It is easy to see that for this case the sum-constraint for the corresponding problem with $J=J_1$ must be tight, otherwise, one can increase the power of the subchannels in the active set to achieve high objective value without violates any constraint. Denote the solution found after Step 6 of \textit{the first iteration} as $p_1,p_2,\cdots,p_{J_1}$. We thus have $\sum_{k=1}^{J_1}p_{k}=P_{J_1}$.

Now we consider the second iteration where $J_2$ is found. If $f'^{(J_2)}=0$, we have $f'^{(J_1)}\ge f'^{(J_2)}$, which ends the proof. If $f'^{(J_2)}>0$, similarly, we conclude that the sum-constraint for the corresponding problem is tight, i.e., $\sum_{k=J_1+1}^{J_2}p_{k}=P_{J_2}-P_{J_1}$ for $p_{J_1+1,\cdots,P_{J_2}}$ found after Step 6 of the 2nd iteration. Then we look at the following 3 optimization problems.
\begin{align}
& {\rm P4.1}: \hspace{-1mm} \max_{p_{1},\cdots,p_{J_1}} \hspace{-0mm}{\sum}_{k=1}^{J_1} f_k(p_k) \nonumber \\
& {\rm{s.t.}}  {\sum}_{k=1}^{J_1} p_{k} \le P_{J_1}, \gamma_k \le p_k \le \tau_k, k=1,\cdots,J_{1}.\\
& {\rm P4.2:} \hspace{-1mm}  \max_{p_{J_1+1},\cdots,p_{J_2}} {\sum}_{k=J_1+1}^{J_2} f_k(p_k) \nonumber \\
  &{\rm{s.t.}}   {\sum}_{k=J_1+1}^{J_2} p_{k} \le P_{J_2}-P_{J_1}, \nonumber\\ & \gamma_k \le p_k \le \tau_k, \ k=J_1+1,\cdots,J_2.\\
& {\rm P4.3:} \hspace{-1mm}  \max_{p_{1},\cdots,p_{J_2}}  {\sum}_{k=1}^{J_2} f_k(p_k) \nonumber \\
 & {\rm{s.t.}}  {\sum}_{k=1}^{J_2} p_{k} \le P_{J_2}, \gamma_k \le p_k \le \tau_k, \ k=1,\cdots,J_2.
\end{align}
The above analysis shows that 1) $p_1,\cdots,p_{J_1}$ is the optimal solution of P4.1 and $\sum_{k=1}^{J_1}p_{k}=P_{J_1}$; 2) $p_{J_1+1},\cdots,p_{J_2}$ is the optimal solution of P4.2 and $\sum_{k=J_1+1}^{J_2}p_{k}=P_{J_2}-P_{J_1}$; 3) $p_{1,J_2},p_{2,J_2},\cdots,p_{J_2,J_2}$ is the optimal solution of P4.3; and 4) $f'^{(J_1)}>f'^{(J_2,old)}$, i.e,  the increasing rate of the active subchannels for P4.1 is larger than that in P4.3.
If $f'^{(J_1)}<f'^{(J_2)}$, from 4), we have $f'^{(J_2)}>f'^{(J_2,old)}$.
Thus, by the concave property of $f_k$, for active subchannels $k_1\in\mathcal{S}_{1,active}$ and $k_2\in\mathcal{S}_{2,active}$, we have $p_{k_1}< p_{k_1,J_2}$ and $p_{k_2}< p_{k_2,J_2}$. For $k\in\mathcal{S}_{1,u}\cup \mathcal{S}_{2,u}$, subchannels that achieve their upper bounds in P4.1 or P4.2, from the result of Lemma \ref{lemma-4}, their increasing rates are higher than the rate of the active ones for P4.1 and P4.2, thus higher than the rate of the active ones in P4.3. Thus $k\in \mathcal{S}_{J_2,u,old}$, meaning these subchannels must achieve their power upper bounds, i.e., $p_{k,J_2}=\tau_k$. For $k\in\mathcal{S}_{1,l}\cup \mathcal{S}_{2,l}$, as $p_k$ achieves its lower bound, $p_{k,J_2}\ge \gamma_k=p_k$ must hold. Thus we have proved
\[
\sum_{k=1}^{J_2}p_{k,J_2}>\hspace{-4mm}\sum_{k\in\mathcal{S}_{1,l}\cup\mathcal{S}_{2,l}}\hspace{-4mm}p_{k}
+\hspace{-2mm}\sum_{k\in\mathcal{S}_{1,u}\cup\mathcal{S}_{2,u}}\hspace{-2mm}p_{k}
+\hspace{-2mm}\sum_{k\in\mathcal{S}_{1,active}\cup \mathcal{S}_{1,active}} \hspace{-4mm}p_k=P_{J_2},
\]
which violates the sum-power constraint of P4.3. This shows that the assumption $f'^{(J_1)}<f'^{(J_2)}$ cannot be true, which ends the proof.
\end{proof}

\subsection{Problems with Multiple Water-Levels}
\label{sect_multiple_level}
For MIMO-OFDM systems, some optimization problems aim at maximizing sum-utilities but with certain levels of fairness among subcarriers, e.g., maximizing the minimum sum-utility function among subcarriers. With the optimal diagonalizable structures, the resource allocation along the eigenchannels aligning the optimal spatial basis can be cast as:
\begin{align}
{\rm P5}: \hspace{-3mm}  \max_{p_{1,j},\cdots,p_{K,j}}\hspace{-1mm} \min_j\ & {\sum}_{k=1}^K  f_{k,j}(p_{k,j})\nonumber \\
  {\rm{s.t.}} \ \ &  {\sum}_{j=1}^J {\sum}_{k=1}^Kp_{k,j} \le P, \ \ p_{k,j} \ge0,
\end{align}where $f_{k,j}$'s have the same properties as the $f_k$-functions in P2.
It can be easily shown that P5 is equivalent to
\begin{align}
\label{p5-variant}
\max_{p_{1,j},\cdots,p_{K,j},t} \  & {\sum}_{k=1}^K f_{k,j}(p_{k,j}) \nonumber \\
{\rm{s.t.}}\ & {\sum}_{k=1}^K f_{k,1}(p_{k,1})=\cdots={\sum}_{k=1}^K f_{k,J}(p_{k,J})=t \nonumber \\
 &  {\sum}_{j=1}^J {\sum}_{k=1}^Kp_{k,j} \le P, \  p_{k,j} \ge0.
\end{align}This is because at the optimum, the objective values on all subcarriers are the same. Define $\mathcal{S}_{j,active}\triangleq\{p_{k,j}>0\}$ for $j=1,\cdots,J$.  The following lemma gives the necessary and sufficient conditions of the optimal solution of P5.
\begin{lemma} The following conditions are both necessary and sufficient for the optimal solution of P5:
\begin{equation}
\left\{ \begin{array}{ll} f_{k_1,j}'(p_{k_1,j})\hspace{-2mm}=\hspace{-2mm}f_{k_2,j}'(p_{k_2,j}) \text{ for }k_1,k_2\in\mathcal{S}_{j,active},\ \hspace{-2mm} j=1,\cdots J,\\
f_{k_1,j}'(p_{k_1,j} \hspace{-1mm}= \hspace{-1mm}0)\hspace{-1mm}\le\hspace{-2mm}f_{k_2,j}'(p_{k_2,j}) \text{ for }  \hspace{-1mm} k_1\notin\mathcal{S}_{j,active}, \hspace{-1mm} k_2\in\mathcal{S}_{j,active},\\
{\sum}_{k=1}^K f_{k,1}(p_{k,1})=\cdots={\sum}_{k=1}^K f_{k,J}(p_{k,J})=t,\\
{\sum}_{j=1}^J {\sum}_{k=1}^Kp_{k,j}=P.
\end{array}
\right. \label{P5-cond}
\end{equation}
\end{lemma}
\begin{proof}
The necessity of the third and fourth lines in (\ref{P5-cond}) is obvious. To show the necessity of the first and second ones, denote the optimal solution of P5 as $p_{k,j}^*$. It is obvious that $\{p_{1,j}^*,\cdots,p_{K,j}^*\}$ must be the optimal solution for the following reduced problem:
\begin{align}
{\rm P5.}j: \ \hspace{-4mm} \max_{p_{1,j},\cdots,p_{K,j}}   & {\sum}_{k=1}^K f_{k,j}(p_{k,j}) \nonumber \\
{\rm{s.t.}}\  &  {\sum}_{k=1}^Kp_{k,j} \le {\sum}_{k=1}^Kp_{k,j}^*, \ p_{k,j} \ge0
\end{align}
since P5.$j$ is a subproblem. From Lemma \ref{lemma-1}, the necessity of the first and second lines of (\ref{P5-cond}) is proved. In P5.$j$, it is obvious that ${\sum}_{k=1}^K f_{k,j}(p_{k,j})$ is monotonic with respect to ${\sum}_{k=1}^Kp_{k,j}$. The solution of $p_{k,j}$'s satisfying (\ref{P5-cond}) have the same objective value for (\ref{p5-variant}) as otherwise the third and fourth lines of (\ref{P5-cond}) cannot hold simultaneously. Then the proof is finished.
\end{proof}

\begin{algorithm}[t]
\caption{Proposed algorithm for problems P5, P6, and P7.}
\label{alg-P5}
\begin{algorithmic}[1]
\STATE $\text{Initialize}$ $\mathcal{I}_{k,j}=1$  $\forall k,j$;
\STATE Calculate $\mu_j$'s and $p_{k,j}$'s by jointly solving  Eqns.~(\ref{P5-1}) and (\ref{P5-3});
\WHILE{$\text{length}( \text{find}(\{p_{k,j}\} < 0 )) > 0$}
\STATE $\mathcal{S}_{j,inactive}=\{(k,j)|p_{k,j} \le 0\}$;
\STATE $\text{Set}$ ${\mathcal{I}}_{k,j}=0$ $\text{for}$ $(k,j)\in \mathcal{S}_{j,inactive}$;
\STATE  Calculate $\mu_j$'s and $p_{k,j}$'s  by jointly solving Eqns.~(\ref{P5-1}) and (\ref{P5-3});
\ENDWHILE
\RETURN $p_{k,j}$'s.
\end{algorithmic}
\end{algorithm}Thus following the viewpoint and scheme in Section \ref{sec-2-2}, we define $g_{k,j}(\cdot)\triangleq {\rm Inv}[f_{k,j}'](\cdot)$ and introduce the indication operator $\mathcal{I}_{k,j}$ as: ${\mathcal{I}}_{k,j}=1$ if $p_{k,j}>0$ and ${\mathcal{I}}_{k,j}=0$ if $p_{k,j}=0$. The necessary conditions can be rewritten as
\begin{align}
&p_{k,j}=g_{k,j}(\mu_j){\mathcal{I}}_{k,j}, \ {\sum}_{k=1}^K f_{k,j}(p_{k,j})= t, \label{P5-1}\\
&P={\sum}_{j=1}^J{\sum}_{k=1}^Kg_{k,i}(\mu_j){\mathcal{I}}_{k,j},\label{P5-3}
\end{align}
where $\mu_j$ is the final increasing rate (also referred to as the water level) for the $j$th subcarrier.
\textbf{Algorithm~\ref{alg-P5}} is thus proposed whose solution values satisfy the necessary and sufficient conditions. Similar to previous discussions, the complexity of Algorithm~\ref{alg-P5} is ${\mathcal{O}}(J^2K^2)$. In the following, two application examples are given to demonstrate the application of Algorithm \ref{alg-P5} for optimization problems in wireless communications.

\noindent \textbf{Example 9:} With the optimal  diagonalization,
the maximum-sum-weighted-MSE minimization problem for MIMO-OFDM systems can be formulated as
\begin{align}
\max_{p_{1,j},\cdots,p_{K,j}} \min_{j} \ &-{\sum}_{k=1}^K\frac{  w_{k,j}}{b_{k,j}+a_{k,j}p_{k,j}} \nonumber \\
{\rm{s.t.}}\ \ \ & {\sum}_{j=1}^J {\sum}_{k=1}^K p_{k,j}\le P,
\end{align}where the term $\sum_{k=1}^K{  w_{k,j}}/({b_{k,j}+a_{k,j}p_{k,j}})$ is the weighted sum-weighted-MSE on the $j$th subcarrier. From (\ref{P5-1}),
\begin{align}
\label{example_9_equ_1}
p_{k,j}=\hspace{-2mm}g_{k,j}(\mu_j)\mathcal{I}_{k,j}=(\sqrt{{w_{k,j}}/
{\mu_ja_{k,j}}}-{b_{k,j}}/{a_{k,j}})\mathcal{I}_{k,j}.
\end{align} By using (\ref{example_9_equ_1}) in (\ref{P5-1}), we have
\begin{align}
\sum_{k=1}^K\frac{w_{k,j}}{b_{k,j}+a_{k,j}p_{k,j}}
&=\!\!\sum_{k=1}^{K}\!\!\frac{w_{k,j}}{b_{k,j}}(1-{\mathcal{I}}_{k,j})+\!\!
\sum_{k=1}^K\!\!{\sqrt{\frac{\mu_jw_{k,j}}{a_{k,j}}}} {\mathcal{I}}_{k,j}\nonumber \\
&=t,
\end{align}from which we can solve $\mu_j$ as a function of $t$ as follows:
\begin{align}
\label{example_9_equ_2}
\sqrt{{1}/{\mu_j}}=\frac{
{{\sum}_{k=1}^K\hspace{-1mm}\sqrt{w_{k,j}/a_{k,j}}}\mathcal{I}_{k,j}}{t-\hspace{-1mm}{\sum}_{k=1}^K\hspace{-1mm}w_{k,j}/b_{k,j}(1-\mathcal{I}_{k,j})}.
\end{align}With (\ref{example_9_equ_2}), the sum power constraint (\ref{P5-3}) can be rewritten as
\begin{align}
\label{example_9_equ_3}
\sum_{j=1}^J\frac{(\sum_{k=1}^K\hspace{-1mm}\sqrt{\frac{w_{k,j}}{{a_{k,j}}}}\mathcal{I}_{k,j})^2}
{t-\sum_{k=1}^K\hspace{-1mm}\frac{w_{k,j}}{b_{k,j}}(1-\mathcal{I}_{k,j}) }
-\sum_{j=1}^J\hspace{-1mm}\sum_{k=1}^K\hspace{-1mm}
\frac{b_{k,j}}{a_{k,j}}{\mathcal{I}}_{k,j}=P.
\end{align}For given ${\mathcal{I}}_{k,j}$'s, bisection search can be used to compute $t$ from (\ref{example_9_equ_3}). Then with (\ref{example_9_equ_2}), $\mu_j$'s can be computed and with  (\ref{example_9_equ_1}), $p_{k,j}$'s can be computed.

\noindent \textbf{Example 10:}
The maximization of the minimum weighted mutual information for MIMO-OFDM systems with
optimal diagonalization can be formulated as the following:
\begin{align}
\max_{p_{1,j},\cdots,p_{K,j}}\min_{j}  \ & {\sum}_{k=1}^K w_{k,j}{\rm{log}}|{b_{k,j}+a_{k,j}p_{k,j}}| \nonumber \\
{\rm{s.t.}}\ &  {\sum}_{j=1}^J{\sum}_{k=1}^K p_{k,j}\le P, \ \ p_{k,j} \ge 0.
\end{align}This optimization plays a key role in the transceiver optimizations of MIMO-OFDM systems with nonlinear Tomlinson-Harashima precoding (THP) or decision feedback equalization.
Following the same logic as that for Example 9, we have via manipulating (\ref{P5-1})
\begin{align}
\label{example_10_equ_1}
&p_{k,j}\hspace{-2mm}=\hspace{-2mm}\left(\frac{w_{k,j}}{\mu_j}
-\frac{b_{k,j}}{a_{k,j}}\right){\mathcal{I}}_{k,j}.\\
&{\rm{log}}\frac{1}{\mu_j}=\frac{t}{{\sum_{k=1}^K}\hspace{-1mm}w_{k,j}{\mathcal{I}}_{k,j}} \nonumber \\
&-\frac{\sum_{k=1}^K\hspace{-1mm}w_{k,j}{\rm{log}}b_{k,j}(1-\mathcal{I}_{k,j})-\hspace{-1mm}
{\sum_{k=1}^K}\hspace{-1mm}w_{k,j}{\rm{log}}(a_{k,j}w_{k,j})\mathcal{I}_{k,j}}
{{\sum_{k=1}^K}\hspace{-1mm}w_{k,j}{\mathcal{I}}_{k,j}}. \label{P5-52}
\end{align}By using the above two equations in (\ref{P5-3}), the value of $t$ can be computed by solving the following equation:
\begin{align}
{\sum}_{j=1}^J{\sum}_{k=1}^K\left({w_{k,j}}/{\mu_j}
-{b_{k,j}}/{a_{k,j}}\right){\mathcal{I}}_{k,j}=P.
\label{P5-53}
\end{align}Because of the monotonicity of the left-hand-side of (\ref{P5-53}) with respect to $t$, bisection search can be used to find the value of $t$. The values of $\mu_j$'s can be obtained from (\ref{P5-52}), and the values of $p_{k,j}$'s can be obtained from (\ref{example_10_equ_1}).

In P5, only the zero lower bound is considered for $p_{k,j}$. It can be generalized to include box constraints. The new optimization problem is given in the following:
\begin{align}
\label{mix_opt}
{\rm P5.1}: \ \max_{p_{1,j},\cdots,p_{K,j}} \min_j \ &{\sum}_{k=1}^K f_{k,j}(p_{k,j}) \nonumber \\
{\rm{s.t.}} \ & {\sum}_{j=1}^J {\sum}_{k=1}^Kp_{k,j} \le P \nonumber \\
  &\gamma_{k,j} \le p_{k,j} \le \tau_{k,j}, p_{k,j} \ge0.
\end{align}An algorithm can be designed based on the combination of \textbf{Algorithms \ref{alg-P4}} and \textbf{\ref{alg-P5}}.
Specifically, an algorithm for P5.1 can be formed by changing ``one of \textbf{Algorithms \ref{alg-P3-a}-\ref{alg-P3-d}}'' in Steps 2,8, 9 of \textbf{Algorithm \ref{alg-P4}}  to ``\textbf{Algorithms \ref{alg-P5}}''.

\noindent \textbf{Example 11:}
The minimization of the maximum weighted MSE under box constraints for MIMO-OFDM systems can be cast as P5, where
\begin{align}
f_{k,j}(p_{k,j})={  w_{k,j}}/({{b_{k,j}+a_{k,j}p_{k,j}}}).
\end{align}A solution can be found with the proposed mixed algorithm where intermediate calculations are given in (\ref{example_9_equ_1})-(\ref{example_9_equ_2}).

\noindent \textbf{Example 12:}
Another example is the combination of Examples 6 and 10 for the minimum capacity maximization in MIMO-OFDM systems, where
\begin{align}
f_{k,j}(p_{k,j})=w_{k,j}{\rm{log}}|{b_{k,j}+a_{k,j}p_{k,j}}|.
\end{align}
A solution can be found with the proposed mixed algorithm where intermediate calculations are given in (\ref{example_10_equ_1})-(\ref{P5-53}).

\subsection{Problems with Aggregation Structures}
\label{sect_robust_wf}

Another set of resource allocation problems have aggregation clustered water-filling structures, e.g., the power allocation in MIMO-OFDM systems under imperfect CSI. With the optimal diagonalizable structure, this type of optimization can be formulated and/or transformed in the following form \cite{ClusterWF}:
\begin{align}
 \hspace{-4mm}{\rm P6}:  \hspace{-2mm} \max_{\{p_{k,j}\},\{P_j\}}   & {\sum}_{j=1}^J{\sum}_{k=1}^K f_{k,j}(p_{k,j},P_j) \nonumber \\
 \ {\rm{s.t.}} \  & {\sum}_{k=1}^K p_{k,j} \le P_j, {\sum}_{j=1}^JP_j \le P, \ p_{k,j} \ge0,
\end{align}where $p_{k,j}$ is the power for Channel/Cluster $j$ on Subcarrier $k$ and $\{P_1,\cdots,P_J\}$ is a set of auxiliary variables representing the total powers over the channels across subcarriers. One important difference of P6 to the previous problems lies in the structure of $f_{k,j}$. Other than $p_{k,j}$'s, it is also a function of $P_j$. With respect to $p_{k,j}$ while $P_j$ is considered fixed, $f_{k,j}$ is assumed to have the same properties (strictly concave, increasing, and continuously differentiable) as before.

It is obvious that for given $P_j$'s, the optimization problem P6 decouples into $J$ subproblems, one for each $J$ and all following the format of P1. Thus similar to Section \ref{sec_new_understanding}, we introduce  the index operator as: ${\mathcal{I}}_{k,j}=1$ if $p_{k,j}>0$ and ${\mathcal{I}}_{k,j}=0$ if $p_{k,j}=0$. From the KKT conditions,
we have
\begin{align}
\label{robust_water_1}
&p_{k,j}=g_{k,j}(\mu_j,P_j){\mathcal{I}}_{k,j}, \
P_{j}={\sum}_{k=1}^Kg_{k,j}(\mu_j,P_j){\mathcal{I}}_{k,j}.
\end{align}Each subproblem can be found via Algorithm \ref{algorithm_1}. To solve $P_j$'s, from KKT conditions of P6, we have\begin{align}
\label{robust_water_3}
&{\sum}_{k=1}^K {\partial f_{k,j}(p_{k,j},P_j)}/{\partial P_j}=\mu_j-\gamma, \\ \label{robust_water_4}
&P={\sum}_{j=1}^J{\sum}_{k=1}^Kg_{k,j}(\mu_j,P_j){\mathcal{I}}_{k,j},
\end{align}where $\mu_j$'s are the Lagrange multipliers for the first set of constraints of P6 and $\lambda$ is the multiplier for the second constraint. Based on (\ref{robust_water_1})-(\ref{robust_water_4}), by following the framework in Section \ref{sec-2-2}, \textbf{Algorithm~\ref{alg-P5}} can be used to find the solution values of P6 as well by only replacing the equations in Steps 2 and 6 to (\ref{robust_water_1})-(\ref{robust_water_4}). Two application examples for P6 are the weighed sum-rate maximization  for MIMO-OFDM systems under imperfect CSI \cite{Xingchinascience} and the weighted sum-MSE minimization \cite{ClusterWF}. Both can be transformed into P6 with the optimal diagonalization structure. The optimal solution can be found via \textbf{Algorithm~\ref{alg-P5}}, where more details are omitted here and interested readers are referred to \cite{ClusterWF}. In the following, only the corresponding objective functions are given.

\noindent \textbf{Example 13:} With optimal diagonalization structure, the problem of the weighted sum-rate maximization for MIMO-OFDM systems under imperfect CSI can also be formulated as P6 \cite{Xingchinascience} where
\begin{align}
\label{robust_equ_a}
 \hspace{-2mm}f_{k,j}(p_{k,j}, P_j)={w_{k,j}}{\rm{log}}\left(1+{a_{k,j}
p_{k,j}}/({{\sigma_{e_j}^2P_j+\sigma_{n}^2}})\right).
\end{align} By jointly solving (\ref{robust_water_1})-(\ref{robust_water_4}), the values of $p_{k,j}$'s can be obtained for any given $\mathcal{I}_{k,j}$.

\noindent \textbf{Example 14:} Similarly with optimal diagonalization structure, the problem of the weighted sum-MSE minimization for MIMO-OFDM systems under imperfect CSI can also be formulated as P6 \cite{ClusterWF} where
\begin{align}
\label{robust_equ_b}
f_{k,j}(p_{k,j}, P_j)=-\frac{w_{k,j}}{1+{a_{k,j}
p_{k,j}}/({\sigma_{e_j}^2P_j+\sigma_{n}^2})}.
\end{align} The values of $p_{k,j}$'s can be computed for any given $\mathcal{I}_{k,j}$, by jointly solving (\ref{robust_water_1})-(\ref{robust_water_4}).

\subsection{Problems with Aggregation Structures and Multiple Water Levels}
A combination of P5 and P6 can be formulated as follows:
\begin{align}
{\rm P7}:\ \max_{\{p_{k,j}\},\{P_j\}} \min_j \ & {\sum}_{k=1}^K f_{k,j}(p_{k,j},P_j) \nonumber \\
{\rm{s.t.}}\  & {\sum}_{k=1}^K p_{k,j} \le P_j \nonumber \\
& {\sum}_{j=1}^JP_j \le P,\ p_{k,j} \ge0,
\end{align}where $f_{k,j}$'s have the same properties as in P6. The problem models the optimization of MIMO-OFDM systems under imperfect CSI with consideration of fairness.
By following the derivations in Sections \ref{sect_multiple_level} and \ref{sect_robust_wf}, we have, from the following KKT conditions,
\begin{align}
&p_{k,j}=g_{k,j}(\mu_j,P_j){\mathcal{I}}_{k,j},  \\
&P_{j}={\sum}_{k=1}^Kg_{k,j}(\mu_j,P_j){\mathcal{I}}_{k,j}. \label{P7-1} \\
& {\sum}_{k=1}^K\hspace{-1mm} f_{k,j}(p_{k,j},P_j)= t, \\
&P={\sum}_{j=1}^J{\sum}_{k=1}^Kg_{k,j}(\mu_j,P_j){\mathcal{I}}_{k,j}. \label{P7-4}
\end{align}\textbf{Algorithm~\ref{alg-P5}} can be used to find the solution of P6 as well by only replacing the equations in Steps 2 and 6 with (\ref{P7-1})-(\ref{P7-4}).

\noindent \textbf{Example 15:} With optimal diagonalization structure, the problem of the minimum weighted sum-rate maximization for MIMO-OFDM systems under imperfect CSI can also be formulated as P7 with (\ref{robust_equ_a}).
By jointly solving (\ref{P7-1}) and (\ref{P7-4}), the values of $p_{k,j}$'s can be obtained for any given $\mathcal{I}_{k,j}$. Thus Algorithm \ref{alg-P5} can be used to find the optimal solutions.

\noindent \textbf{Example 16:} On the other hand, with optimal diagonalization structure, the problem of the maximum weighted sum-MSE minimization for MIMO-OFDM systems under imperfect CSI can also be formulated as P7 with (\ref{robust_equ_a}). Similar to Example 15,
 the values of $p_{k,j}$'s can be obtained for any given $\mathcal{I}_{k,j}$ via jointly solving (\ref{P7-1}) and (\ref{P7-4}). Thus Algorithm \ref{alg-P5} still works.

\subsection{Problems with Subchannel Clsuters/Groups}

In \cite{PeterHe2017,PeterHeTSG2018,PeterHeTWC2018}, water-filling algorithms are designed to implement over subchannel groups (clusters of subchannels) with individual group/cluster constraints.
In the following we show that the proposed water-filling algorithm is applicable to this case.
When there are several disjoint subchannel groups ${\mathcal{C}}_k$'s, and for each subchannel group  there is a minimum power constraint, the corresponding optimization under a set of group lower bound constraints is given as follows
\begin{align}
\label{opt_lower_cluster}
 \hspace{-1mm} {\rm P8:}  \hspace{-1mm} \max_{p_1,\cdots,p_K} \  &{\sum}_{k=1}^K f_k(p_k) \nonumber \\
{\rm{s.t.}} \ & {\sum}_{k=1}^K p_{k} \le P,{\sum}_{k \in \mathcal{C}_j} p_k \ge \gamma_j, p_k\ge 0.
\end{align}The constraint $p_k\ge0$ must be taken into account to guarantee that the final solution is of practical meanings. Actually from the mathematical viewpoint the optimization problem without $p_k\ge0$ is much easier to handle. In our work, we consider a general case that ${\mathcal{C}}_1\cup\cdots\cup {\mathcal{C}}_J \subseteq\{1,\cdots,K\}$ and ${\mathcal{C}}_n\cap{\mathcal{C}}_m=\emptyset$ for $n\not= m$. In other words, it is possible that there are some subchannels that do not belong to any of the subchannel groups. Comparing P1.1 and P8, it is obvious that  in P8 the constraints are added on subchannel groups instead of each subchannel. Each group ${\mathcal{C}}_k$ can be regarded as a virtual subchannel for which there are two kinds of constraints, i.e., the power allocated to each subchannel must be nonnegative and the sum power of its subchannels should be larger than a threshold. Inspired by this idea, to solve the optimization problem (\ref{opt_lower_cluster}), Algorithm~\ref{alg-2a} is proposed. Comparing Algorithm~\ref{alg-2} and Algorithm~\ref{alg-2a}, it can be concluded that the judgement step for each suchannel constraint is replaced by the judgement step for each subchannel group constraint. In other words, these new operations are defined based on the groups of subchannels. In each iteration, the computation of $\mu$ and $p_k$  in Algorithm~\ref{alg-2} is replaced by performing \textbf{Algorithm 1} in Algorithm~\ref{alg-2a}.
\begin{algorithm}[t]
\caption{Proposed algorithm under arbitrary group lower-bound constraints for P8.}
\label{alg-2a}
\begin{algorithmic}[1]
\STATE $\mathcal{I}_j=1$  $\text{for}$ $j=1,\cdots,J$;
\STATE $\text{Perform}$ $\textbf{Algorithm}~\ref{algorithm_1}$  over all subchannels with sum power $P$ ;
\WHILE{$\text{length}( \{j|\sum_{k\in {\mathcal{C}}_j}p_k < \gamma_j \})>0$}
\STATE Find $\mathcal{S}_{L}=\{j|\sum_{k\in {\mathcal{C}}_j}p_k \le \gamma_j\}$;
\STATE $\text{Set}$ ${\mathcal{I}}_j=0$ $\text{for}$ $ j \in \mathcal{S}_{L}$;
\STATE $\text{Perform}$ $\textbf{Algorithm~\ref{algorithm_1}}$ $\text{over}$ the subchannel set  $\{k|k\notin \mathcal{C}_{j}, j \in \mathcal{S}_{L} \}$ with power $P-\sum_{j=1}^J\gamma_j(1-{\mathcal{I}}_j)$ ;
\ENDWHILE
\STATE  $\text{Perform}$ $\textbf{Algorithm~\ref{algorithm_1}}$ $\text{over}$ $\mathcal{C}_j$, $ j \in \mathcal{S}_{L}$ with power $\gamma_j$ ;
\RETURN $p_k$'s.
\end{algorithmic}
\end{algorithm}
\begin{algorithm}[t]
\caption{Proposed algorithm under arbitrary group box constraints for P9.}
\label{alg-2b}
\begin{algorithmic}[1]
\STATE $\mathcal{J}_j=1$  $\text{for}$ $j=1,\cdots,J$;
\STATE $\text{Perform}$ $\textbf{Algorithm~\ref{alg-2a}}$ $\text{over}$ all subchannels with sum power $P$ ;
\WHILE{$\text{length}( \{j|\sum_{k\in {\mathcal{C}}_j}p_k > \tau_j\})>0$}
\STATE Find $\mathcal{S}_{U}=\{j|\sum_{k\in {\mathcal{C}}_j}p_k \ge \tau_j\}$;
\STATE $\text{Set}$ ${\mathcal{J}}_j=0$ $\text{for}$ $ j \in \mathcal{S}_{U}$;
\STATE $\text{Perform}$ $\textbf{Algorithm~\ref{alg-2a}}$ $\text{over}$ the subchannel set  $\{k|k\notin \mathcal{C}_{j}, j \in \mathcal{S}_{U} \}$ with power $P-\sum_{j=1}^J\tau_j(1-{\mathcal{J}}_j)$ ;
\ENDWHILE
\STATE  $\text{Perform}$ $\textbf{Algorithm~\ref{alg-2a}}$ $\text{over}$ $\mathcal{C}_j$, $ j \in \mathcal{S}_{U}$ with power $\tau_j$ ;
\RETURN $p_k$'s.
\end{algorithmic}
\end{algorithm}

Moreover, we take a step further to investigate a more complicated optimization problem with subchannel groups under box constraints. The corresponding mathematical formula is given in the following
\begin{align}
\label{opt_lower_upper_cluster}
 \hspace{-4mm}{\rm P9:} \hspace{-2mm}  \max_{p_1,\cdots,p_K} \  &{\sum}_{k=1}^K f_k(p_k) \nonumber \\
{\rm{s.t.}} & {\sum}_{k=1}^K p_{k} \le P,\gamma_j \le {\sum}_{k \in \mathcal{C}_j} p_k\le \tau_j, p_k\ge 0,
\end{align} which is an extension of P3. Similarly, Algorithm~\ref{alg-2b} is proposed. The design logic is similar to Algorithm~\ref{alg-P3-a} and  Algorithm~\ref{alg-2b} can be recognized as an extension. In a nutshell, the proposed water-filling algorithm design framework can be applied to solve the optimization problem.

\section{Numerical Results}
\label{sect_numerical_result}

Due to space limitations, only the most general and representative cases are simulated and shown.
In most existing works, sum capacity maximization is the main focus. It is worth noting that at high signal-to-noise ration (SNR) regime, the optimal power allocation for capacity maximization is close to equal power allocation among all the subchannels or subcarriers, making this scenario less interesting. However, the optimal solution for sum MSE minimization is totally different. In the high SNR regime, the powers allocated to different subchannels are significantly different from each other. Moreover, MIMO-OFDM system is a well-known communication system and has been adopted
in 4G-LTE. With these considerations, sum-MSE minimization for MIMO-OFDM system is investigated in this section.

Specifically, sum-MSE minimization for MIMO-OFDM systems with box constraints given in \textbf{Example 7} is considered for illustration. In practice, box constraints are widely imposed to control the fairness among subcarriers/subchannels and/or limit the peak-to-average-power-ratio (PAPR) in OFDM systems. The sum-MSE minimization problem in fact can be formulated as Problem P3 in (\ref{Opt_2}).
In the simulations, the number of antennas at both the source and destination are set as $K=4$. A multi-path channel with 7 paths is considered. Each channel tap is generated according to a Gaussian distribution. Further, the time-domain decaying factor of the channel taps is 0.5, i.e., the covariance of the $l$th tap is $1/2^{l-1}$, under which the first tap's covariance and the sum variance of all the taps are normalized to one.
The number of subcarriers is $J=256$. The maximum transmit power is denoted by $P$ and the SNR is defined as $P/(J\sigma_{n}^2)$. The box constraints are considered as $\gamma P/(KJ) \le p_{k,j}\le \tau P/(KJ)$, where $\gamma$ and $\tau$ are introduced to adjust the bounds of the power allocations to exhibit the effect of box constraints. Each point in the figures is obtained by an average over $10^3$ independent channel realizations. The sum-MSE under different SNRs is shown in Fig.~\ref{Fig_1} where it can be observed that the box constraints do matter and would affect the sum-MSE. In Figs.~\ref{Fig_2} and  \ref{Fig_3}, we show the power allocation results across the 1024 eigen-channels for one channel realization at the SNR of 20 dB with different\begin{figure}[t]
\centering
\includegraphics[width=.5\textwidth]{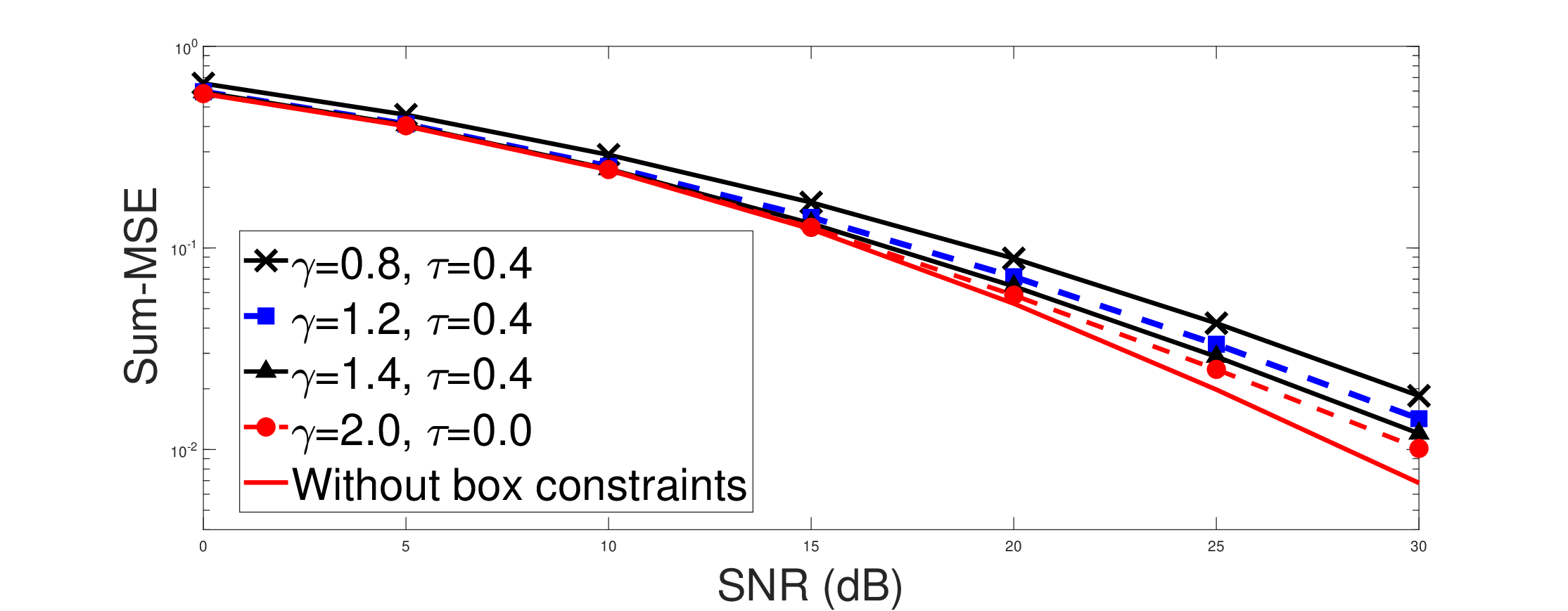}
\caption{The MSEs under different box constraints for single-user MIMO-OFDM system with 4 antennas at transceiver nodes.}\label{Fig_1}
\end{figure}\begin{figure}[t]
\centering
\includegraphics[width=.5\textwidth]{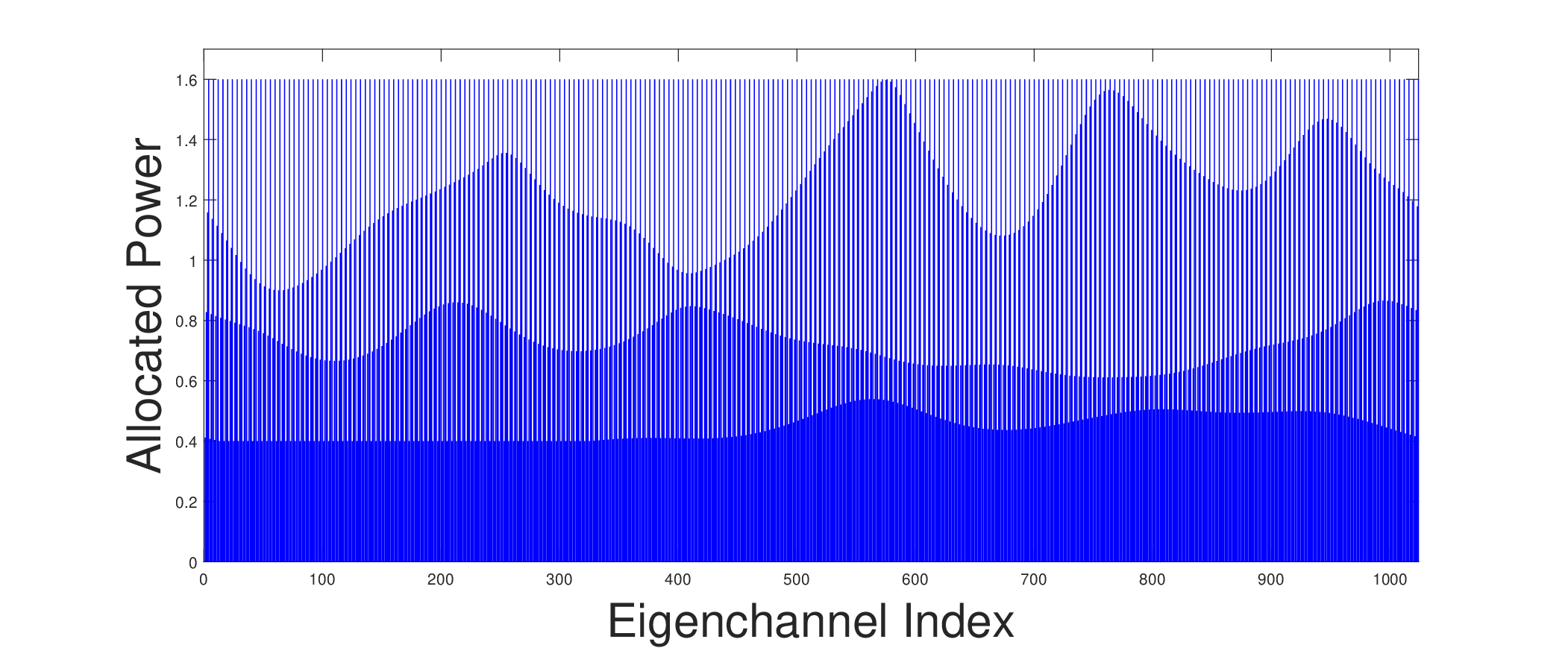}
\caption{The power allocation for one channel realization at the SNR of 20 dB with $\gamma=0.4$ and $\tau=1.6$.}\label{Fig_2}
\end{figure}
\begin{figure}[!ht]
\centering
\includegraphics[width=.5\textwidth]{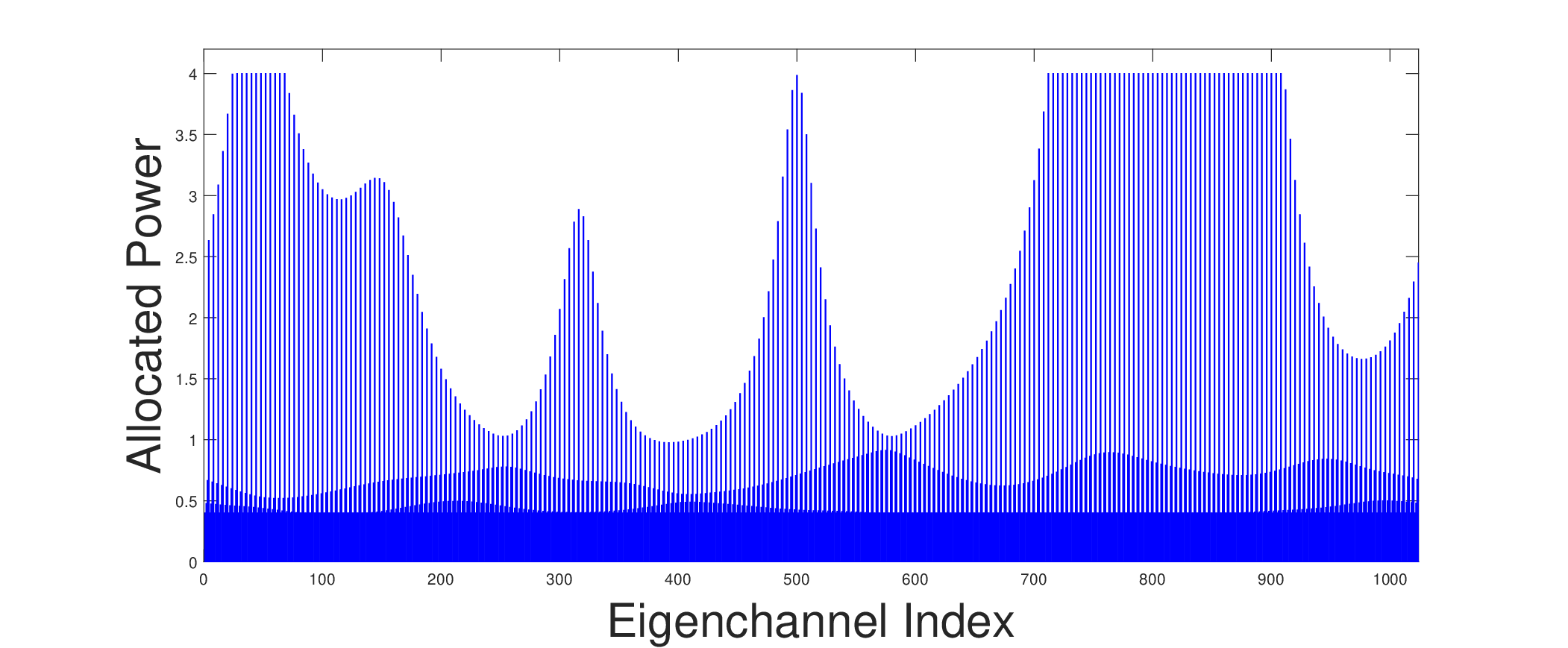}
\caption{The power allocation for one channel realization at the SNR of 20 dB with $\gamma=0.4$ and $\tau=4$.}\label{Fig_3}
\end{figure}values of $\gamma$ and $\tau$. It can be seen that box constraints have significant impact on power allocation as the channel qualities of different eigenchannels fluctuate significantly.
Furthermore, it can be seen that both the lower bounds and upper bounds can be met. Even when the lower bound is very small or the upper bound is very high, the constraints are usually active.


\section{Conclusions}
\label{sect_conclusion}
Optimization problems with water-filling solutions widely arise and are fundamental for the resource allocation in wireless communications and networking. To find the solution values, practical and efficient algorithms should be carefully designed. In this work, a new viewpoint for such optimization problems has been proposed by understanding the power allocation procedure dynamically and considering the changes of the increasing rates on each subchannel. With this viewpoint and rigorous analysis of the solution structure, a comprehensive framework for algorithm designs has been presented in this paper. Five different kinds of optimization problems have been studied sequentially according to their complexities and efficient algorithms have been proposed. Based on our results, it can be concluded that the various algorithm designs share common fundamentals. We also expect that the proposed design logic and algorithms can be used to resolve new optimization problems in future wireless systems.

\section*{Acknowledgement}

The authors sincerely appreciate the anonymous reviewers and the editor for their excellent review works on our paper.

\end{document}